\documentclass[11pt]{article}
\usepackage{float}
\usepackage{amsmath}
\usepackage{amsthm}
\usepackage{amstext}
\usepackage{amsmath}
\usepackage{amssymb}
\usepackage{amsfonts}
\usepackage[]{algorithm2e, algorithmicx, algpseudocode}
\usepackage{color}
\usepackage{cite}
\usepackage{url}
\usepackage{bbm}
\usepackage{enumitem} 
\usepackage[margin=1in]{geometry}
\usepackage{mathdots}
\usepackage{bm}
\RestyleAlgo{boxed}

\newcommand{\prob}[1]           {\Pr\left\{ #1 \right\}}

\theoremstyle{plain}

\newtheorem{theorem}            {Theorem}
\newtheorem*{theorem*}            {Theorem}
\newtheorem{lemma}[theorem]     {Lemma}
\newtheorem*{lemma*}     {Lemma}
\newtheorem{remark}[theorem]     {Remark}

\newtheorem{cclaim}[theorem]     {Claim}
\newtheorem{corollary}[theorem] {Corollary}
\newtheorem{definition}[theorem]{Definition}

\newtheorem{observation}[theorem]     {Observation}

\newcommand{\secref}[1]         {Section~\ref{sec:#1}}
\newcommand{\seclabel}[1]    {\label{sec:#1}}
\newcommand{\figlabel}[1]   {\label{fig:#1}}
\newcommand{\figref}[1]         {Figure~\ref{fig:#1}}
\newcommand{\thmlabel}[1]   {\label{thm:#1}}
\newcommand{\thmref}[1]         {Theorem~\ref{thm:#1}}
\newcommand{\lemlabel}[1]   {\label{lem:#1}}
\newcommand{\lemref}[1]         {Lemma~\ref{lem:#1}}

\newcommand{\corlabel}[1]   {\label{cor:#1}}
\newcommand{\corref}[1]         {Corollary~\ref{cor:#1}}

\renewcommand{\eqref}[1]          {Eq.~\ref{eq:#1}}
\newcommand{\deflabel}[1]    {\label{def:#1}}
\newcommand{\defref}[1]         {Definition~\ref{def:#1}}

\renewcommand{\epsilon}{\varepsilon}

\renewcommand{\sf}[1]{\mathsf{#1}}

\newcommand{\cncRIP}{\sf{O}(1)\mbox{-}\sf{ncRIP}}
\newcommand{\ancRIP}{\sf{\alpha}(n)\mbox{-}\sf{ncRIP}}
\newcommand{\pncRIP}{\sf{poly}(n)\mbox{-}\sf{ncRIP}}
\newcommand{\encRIP}{\sf{exp}(n)\mbox{-}\sf{ncRIP}}
\newcommand{\cMRIP}{\sf{O}(1)\mbox{-}\sf{MRIP}}

\newcommand{\pMRIP}{\sf{poly}(n)\mbox{-}\sf{MRIP}}
\newcommand{\eMRIP}{\sf{exp}(n)\mbox{-}\sf{MRIP}}

\newif\ifnotes
\notestrue
\ifnotes
\newcommand{\jing}[1]{\textcolor{blue}{{\footnotesize #1}}\marginpar{\raggedright\tiny \textcolor{blue}{JING}}}
\definecolor{purple}{rgb}{.5,0,.5}
\newcommand{\sam}[1]{\textcolor{purple}{{\footnotesize #1}}\marginpar{\raggedright\tiny \textcolor{purple}{SAM}}}
\definecolor{orange}{rgb}{1,0.6,0}
\newcommand{\shikha}[1]{\textcolor{orange}{{\footnotesize #1}}\marginpar{\raggedright\tiny \textcolor{orange}{SHIKHA}}}
 
\else
\newcommand{\sam}[1]{}
\newcommand{\shikha}[1]{}
\newcommand{\jing}[1]{}
\fi

\newcommand{\defn}[1]{\textbf{\emph{#1}}}

\title{
  Non-Cooperative Rational Interactive Proofs
\footnote{A preliminary version of this paper appeared at the 27th European Symposium on Algorithms (ESA 2019). 
This work has been partially supported by NSF CAREER Award CCF 1553385,
CNS 1408695, CCF 1439084, IIS 1247726, IIS 1251137, CCF 1217708, by Sandia National Laboratories,
and by the European Research Council under the European Union's 7th Framework Programme (FP7/2007-2013)~/~ERC grant agreement no. 614331.
BARC, Basic Algorithms Research Copenhagen, is supported by the VILLUM Foundation grant 16582.}
}

\author{
Jing Chen\thanks{
Stony Brook University,
Stony Brook, NY 11794-4400, USA.
Email:~\texttt{jingchen@cs.stonybrook.edu}.
}
\and
Samuel McCauley\thanks{
  Williams College, Williamstown MA 01267 USA.
Email:~\texttt{\{sam, shikha\}@cs.williams.edu}.
}
\and
Shikha Singh\footnotemark[3]
}
\date{}

\begin{document}

\sloppy

\maketitle


\begin{abstract}

Interactive-proof games model the scenario where an honest party interacts with powerful but strategic provers, to elicit from them the correct answer to a computational question. 
  Interactive proofs are increasingly used as a framework to design protocols for computation outsourcing.

Existing interactive-proof games largely fall into two categories:
either as games of cooperation such as multi-prover interactive proofs and cooperative rational proofs, where the provers work together
as a team; or as games of conflict such as refereed games, where the provers directly compete with each other in a zero-sum game.
%
%
Neither of these extremes truly capture the strategic nature of
   service providers in outsourcing applications.
   How to design and analyze non-cooperative interactive proofs is an important open problem.  

In this paper, we introduce a mechanism-design approach to define 
a multi-prover interactive-proof model in which the
provers are {\em rational} and {\em non-cooperative}---they act to maximize
their expected utility given others' strategies. 
%
%
We define a strong notion of backwards induction as our solution concept 
to analyze the resulting extensive-form game with imperfect information.


We fully characterize the complexity of our proof system under
different \emph{utility gap} guarantees. (At a high level, a utility gap of $u$ means that the protocol is robust against provers that may not care
about a utility loss of $1/u$.) We show, for example, that the power of non-cooperative rational interactive proofs with a polynomial
utility gap is exactly
equal to the complexity class $\sf{P^{NEXP}}$.
%
%
%
\end{abstract}



\section{Introduction}
\seclabel{intro}
 
Game theory has played a central role in analyzing the conflict and cooperation in interactive proof games. These games model the scenario where an honest party interacts with powerful but strategic agents, to elicit from them the correct answer to a computational question. The extensive study of these games over decades has fueled our understanding of important complexity classes~(e.g.,~\cite{babai1991non,fortnow1994power,feige1992two,lund1992algebraic, chandra1976alternation, feige1997making, feige1992multi,
feigenbaum1995game}). From a modern perspective, these games capture the essence of computation outsourcing---the honest party is a client outsourcing his computation to powerful rational service providers in exchange for money. 

In this paper, we consider a natural type of interactive-proof game. 
For the moment, let us call our client Arthur.
Arthur hires a service provider Merlin to solve a computational problem for him, and hires a second service provider Megan to cross-check Merlin's answer. 
Arthur wants the game (and associated payments) to be designed such that if Merlin gives the correct answer, Megan agrees with him; however, if Merlin cheats and gives a wrong answer, Megan is incentivized to contradict him, informing Arthur of Merlin's dishonesty.
This means that Merlin and Megan are not purely cooperative nor purely competitive.  Each is simply a rational agent who wants to maximize their own utility.  

This is a mechanism design problem---how can Arthur incentivize non-cooperative rational agents (Merlin and Megan) to give truthful answers to his questions, helping him solve a computational problem? This problem is the focus of our paper.

\paragraph*{Structure of the game.}  We borrow the structure and terminology of interactive proofs~\cite{babai1985trading,IP,ben1988multi}, as was done in previous work on rational proofs~\cite{azar2012rational,
azar2013super, 
guo2014rational, 
guo2016rational,
ChenMcSi16, 
ChenMcSi18, 
CMS17arxiv, 
campanelli2015sequentially,
campanelli2017efficient} 
and refereed games~\cite{
chandra1976alternation, 
feige1990noisy,
feige1997making, 
feige1992multi, 
reif1984complexity,
feigenbaum1995game,
koller1992complexity}.  
We call Arthur the \defn{verifier} and assume that he is computationally bounded (he may be probabilistic, but must run in polynomial time).
Arthur's coin flips are treated as Nature moves in the game. We call Merlin and Megan the \defn{provers}; they have unbounded computational power.  

The verifier exchanges messages with the provers in order to determine the answer
to a decision problem. 
%
%
The exchange proceeds in rounds: 
in a round, either a verifier sends a message to all provers or receives a response from each.   
The provers cannot observe the messages exchanged between the verifier and other provers.

At the end, the verifier gives a payment to {\em each} prover.  
Our goal is to design protocols and payments such that, under an appropriate solution concept
of the resulting game, the provers' best strategies lead the verifier to the correct answer. 

The interactive protocols described above form an extensive-form game of imperfect information.
To analyze them, we essentially use a strong notion of backward induction as our solution concept.
We refine it further by eliminating strategies that are weakly dominated on ``subgames'' within the entire game.
We define the solution concept formally in~\secref{sse}.
%
%

\paragraph*{Comparison to previous work.}
The model of our games is based on interactive proof systems~\cite{babai1985trading,IP}, in which a verifier exchanges messages with untrustworty provers
and at the end either accepts or rejects their claim. 
Interactive proofs guarantee that, roughly speaking, the verifier accepts a truthful claim with probability at least 2/3 (\defn{completeness}) and no strategy of the provers can make the verifier accept a false claim with probability more than 1/3 (\defn{soundness}).

%
The study of interactive proofs has found extensive applications in both theory and practice.  
Classical results on IPs have led
us to better understand complexity classes through characterizations such as $\sf{IP}=\sf{PSPACE}$~\cite{Shamir92,lund1992algebraic} and $\sf{MIP = NEXP}$~\cite{babai1991non,fortnow1994power,feige1992two}, and later
led to the important area of probabilistically checkable proofs~\cite{sudan2009probabilistically}.
More recently, the study of IPs has resulted in extremely efficient (e.g., near linear or even logarithmic time) protocols
for delegation of computation~\cite{bitansky2012succinct,
goldwasser2008delegating,
canetti2013refereed,
braun2013verifying,
rothblum2013interactive
}.
Such super-efficient IPs have brought theory
closer to practice, resulting in ``nearly practical''
systems~(e.g., see~\cite{
blumberg2014verifiable,
walfish2015verifying,
thaler2012verifiable,
canetti2011practical
}).


Indeed, interactive proofs are not only a fundamental theoretical concept but an
indispensable framework to design efficient computation-outsourcing protocols.


\paragraph*{Existing interactive-proof games} 
Interactive-proof systems with multiple provers have largely been studied as games that fall into two categories:
either as games of cooperation such as MIP~\cite{ben1988multi}, cooperative multi-prover rational proofs (MRIP)~\cite{ChenMcSi16},
and variants~\cite{fortnow1994power,babai1991non, cai1992games, ito2012multi,
goldwasser2008delegating}, where the provers work together to convince
the verifier of their joint claim; or as games of conflict such as refereed games~\cite{chandra1976alternation, feige1997making, feige1992multi,
feigenbaum1995game,canetti2013refereed, kol2013competing, 
canetti2012two}, where the provers directly compete with each other to convince the verifier
of their conflicting claims.

Both of these categories have limitations.
In a game of cooperation, 
provers cannot be leveraged directly against each other.  That is, the verifier cannot directly ask one prover if another prover is lying.
On the other hand, in a game of conflict, such as refereed games, one prover must ``win'' the zero-sum game.
Thus, such games need to assume that at least one prover---who must be the winning prover in a correct protocol---can be trusted to always tell the truth.
Despite their limitations, both models 
have proved to be fundamental constructs to understand and characterize important complexity classes~\cite{
babai1991non,
feigenbaum1995game, feige1997making, chandra1976alternation, ChenMcSi16}, and to design efficient computation outsourcing protocols~\cite{bitansky2012succinct,
blumberg2014verifiable, goldwasser2008delegating,
canetti2013refereed, 
canetti2012two}.

\subsection{Contributions and Results}
In this paper, we introduce
a new interactive-proof game, {\em non-cooperative rational interactive proofs (ncRIP)}.
This model generalizes multi-prover rational proofs~\cite{ChenMcSi16,CMS17arxiv,ChenMcSi18}. 

\paragraph*{Solution concept for ncRIP}
We define a refinement of sequential equilibrium~\cite{kreps1982sequential},
\defn{strong sequential equilibrium} (SSE), that essentially says that players' beliefs about the histories
that led them to an unreachable information set should be irrelevant to their best response. 
From a mechanism-design perspective, we want to design the protocols and payments
that allow this strong guarantee to hold---letting the players' best responses be unaffected by their beliefs.\footnote{We believe that SSE is of independent interest as a solution concept for designing extensive-form mechanisms (e.g.~\cite{
glazer1996virtual,vartiainen2007subgame,duggan1998extensive}).
In~\secref{strongse}, 
we prove important properties of SSE that may prove useful in future studies.}


Finally, we eliminate SSE strategies that are suboptimal within ``subgames'' by defining and enforcing a backward-induction-compatible
notion of dominance.
%
%
%
Roughly speaking, we say a protocol is a ncRIP if there exists a strategy profile of the provers
that is a dominant SSE among the \emph{subforms} of the extensive form game,
and under this strategy
the provers' lead the verifier to the correct answer.
%
We define the model formally in~\secref{ncmrip}.

\paragraph*{Utility gap for non-cooperative provers}
Utility gap is a fundamental concept for rational proofs~\cite{azar2013super, guo2014rational, ChenMcSi16, ChenMcSi18} which is analogous to {\em soundness gap} in interactive proofs.
It measures how robust a protocol is
against the provers' possible deviations from the desired strategy.

This notion is straightforward to define for cooperative rational protocols---they have
a utility gap of $u$ if the {\em total} expected payment decreases by $1/u$
whenever the provers report the wrong answer.
In non-cooperative protocols, however, it is not a priori clear
how to define such a payment loss or to choose which prover should incur the loss.
A payment loss solely imposed on the total payment
may not prevent some provers from deviating,
and a loss solely imposed on the provers' final payments may not prevent them from deviating within subgames.

We define a meaningful notion of utility gap for ncRIP
that 
is naturally incorporated in a backward-induction-compatible way to the dominant SSE concept.

\paragraph*{Tight characterizations of ncRIP classes}
In this paper, we completely characterize the power of non-cooperative rational proofs
under different utility-gap guarantees. 

We construct ncRIP protocols with
constant, polynomial, and exponential utility gaps for powerful complexity classes,
demonstrating the strength of our solution concept.
Our protocols are simple and intuitive (requiring only a few careful tweaks from their cooperative counterparts), and are thus easy to explain and implement.
However, proving their correctness involves analyzing the extensive-game (including subtleties in the incentives and beliefs of each player at each round)
to show that the protocol meets the strong solution-concept and utility-gap requirements.


We then prove {tight} upper bounds for all three ncRIP classes.
Proving tight upper bounds is the most technically challenging part
of the paper. 
We prove the upper bounds by simulating the decisions of the verifier and provers with a Turing Machine. However, there are several obstacles to attain
the correct bounds.  
For example, the polynomial randomness of the verifier can induce an exponential-sized game tree, which is too large to be verified by the polynomial-time machine in Theorems~\ref{thm:constchar} and~\ref{thm:polychar}. Furthermore, an NEXP oracle cannot itself verify whether a strategy profile is a dominant SSE.    
The key lemma that helps us overcome these challenges is the pruning lemma~(\lemref{pruning}). At a high level, it shows that we can prune the nature moves
of the verifier in the resulting game tree, 
while preserving the dominant-SSE and utility-gap guarantees.

Our results are summarized in~\figref{mainresults}, where
 we use $\cncRIP$, $\pncRIP$ and $\encRIP$ to denote ncRIP classes
with constant, polynomial and exponential utility gaps respectively.
The notations are analogous for MRIP~\cite{CMS17arxiv} (the cooperative variant).
We characterize ncRIP classes via oracle Turing machines.
In particular, $\sf{P^{NEXP[O(1)]}}$ is
the
class of languages decided by a polynomial-time Turing machine that makes $O(1)$ queries
to an $\sf{NEXP}$ oracle,
and $\sf{EXP^{poly\mbox{-}NEXP}}$ is the class decided by
an exponential-time Turing machine with polynomial-length queries to an $\sf{NEXP}$ oracle.
\newlength{\thmfiguretab}
\setlength{\thmfiguretab}{-4pt}

\vspace{-10pt}
\noindent
\begin{figure}[h]
\centering
\fbox{\normalsize
\begin{minipage}{0.48\textwidth}
\begin{theorem}\thmlabel{constchar}
  \hspace{\thmfiguretab}$\cncRIP = \sf{P^{NEXP[O(1)]}}$
\end{theorem}
\begin{theorem}\thmlabel{polychar}
\hspace{\thmfiguretab}$\pncRIP = \sf{P^{NEXP}}$
\end{theorem}
\begin{theorem}\thmlabel{ncmrip-char}
  \hspace{\thmfiguretab}$\encRIP = \sf{EXP^{poly\mbox{-}NEXP}}$
\end{theorem}
\end{minipage}%
\begin{minipage}{0.49\textwidth}
\begin{corollary}\corlabel{c-mrip-ncrip}
\hspace{\thmfiguretab}$\cncRIP=\cMRIP$
\end{corollary}
\begin{corollary}\corlabel{p-mrip-ncrip}
\hspace{\thmfiguretab}$\pncRIP \supseteq \pMRIP$
\end{corollary} 
\begin{corollary}\corlabel{exp-mrip-ncrip}
\hspace{\thmfiguretab}$\encRIP =\eMRIP$
\end{corollary}
\end{minipage}
}\caption{Summary of our results.}\figlabel{mainresults}
\end{figure}
\vspace{-5pt}


%
\paragraph*{Power of non-cooperative vs. cooperative and competitive provers}
Interestingly, in the case of constant and exponential utility gap, the power
of ncRIP and MRIP coincide.
%
This can be explained by the power of adaptive versus non-adaptive queries in oracle Turing machines.

Indeed, our results reveal the main difference between non-cooperative and cooperative provers: the former can be used to handle
adaptive oracle queries, the latter cannot (see~\cite{ChenMcSi16, CMS17arxiv}).
Intuitively, this makes sense---cooperative provers may collude across
adaptive queries, answering some of them incorrectly to
gain on future queries. On the other hand, non-cooperativeness allows us to treat the subgame involving the oracle queries
as a separate game from the rest.
%
%

Our results also show that non-cooperative provers are more powerful than competing provers. Feige and Kilian~\cite{feige1997making} proved that the power of refereed games with imperfect information and perfect
recall is equal to $\sf{EXP}$.
\section{Non-Cooperative Rational Interactive Proofs}\seclabel{ncmrip}
In this section we introduce the model for ncRIP.  
\paragraph*{Notation.}
First, we review the structure of ncRIP protocols and related notation; this is largely the same as~\cite{ChenMcSi16}.

The decision problem being solved by an interactive proof is
modeled as whether a given string $x$ is in language $L$. 
An interactive protocol is a pair $(V, \vec{P})$, where
$V$ is the \defn{verifier}, $\vec{P} = (P_1,\ldots,P_{p(n)})$ is the vector
of $p(n)$ \defn{provers}, where $p(n)$ is polynomial in $n = |x|$.  The verifier runs in polynomial time and flips private coins. Each
$P_i$ is computationally unbounded.  The verifier and provers are given the
input $x$. Similar to classical multi-prover interactive proofs, the verifier can communicate with each prover privately, but no two
provers can communicate with each other once the protocol begins.  

In a \defn{round}, either each prover sends a message to $V$, or $V$ sends a message to each prover, and these two cases
alternate.
The length of each message $\ell(n)$, and the number of rounds $k(n)$ are both polynomial in $n$.
The final
transcript $\vec{m}$ of the protocol
is a random variable depending on $r$, the random string used by $V$.
At the end of the communication, 
the verifier computes an \defn{answer bit} $c\in\{0, 1\}$ for the membership of $x$ in $L$ based on $x$, $r$, and $\vec{m}$.
$V$ also computes a payment vector $\vec{R} = (R_1, R_2,
\ldots, R_{p(n)})$, where $R_i$ is the payment given to $P_i$, $R_i \in [-1,1]$,
and the total $\sum_{i=1}^{p(n)} R_i \in [-1,1]$ as well.\footnote{Negative payments are used to reflect punishment. 
The individual payments and the total payment can be shifted and scaled to lie in $[0,1]$.}
The protocol and the payment function $\vec{R}$ are public knowledge.

Each prover $P_i$'s \defn{strategy} at round $j$ 
maps the transcript seen at the beginning of round $j$ to the message he sends in that round.
Let $s_i =(s_{i1},\ldots, s_{ik(n)})$ be the
strategy of prover $P_i$, and $s = (s_1,\dots, s_{p(n)})$
be the \defn{strategy profile} of the provers.  Given input $x$, 
and strategy profile $s$, 
let $u_k(x, s, (V, \vec{P}))$ denote the expected payment of prover $P_k$
in the protocol $(V, \vec{P})$ based on randomness $r$, input $x$ and $s$; if $(V,\vec{P})$ is clear from context, we shorten this to  $u_k(x,s)$
or $u_k(s)$.

The protocol forms an \defn{extensive-form game with imperfect information} which we
describe in the next section. 
The protocol and payments should be designed such that the provers are incentivized to reach an equilibrium that leads $V$ to the correct answer bit $c$.
We formalize the solution concept in~\secref{sse}.

\subsection{
Extensive-form Games and ncRIP}\seclabel{extensive}

We describe the underlying extensive-form game resulting from ncRIP protocols in this section.
For details on extensive-form  games, we refer to the
textbook by Osborne and
Rubinstein~\cite{osborne1994course}.

In a protocol $(V, \vec{P})$ with input $x$, the set of provers
$\vec{P}=(P_1,\ldots,P_{p(n)})$ are the \defn{players}. $V$ is
not a player of the game---the deterministic moves of $V$ form the structure
of the game tree and the randomized moves of $V$ are treated as \defn{Nature} moves.

A \defn{history} $h$ of the game is a sequence of actions taken by the players, written
$h = (a^1, a^2, \ldots, a^K)$ for some actions $a^1, \ldots, a^K$.
The set of histories (including $\phi$, the empty history corresponding to the root) is denoted by ${H}$.
Note that every prefix of $h = (a^1, a^2, \ldots, a^K) \in H$ must also be a valid history, that is, $(a^1, a^2, \ldots, a^L) \in {H}$ for any $L < K$.

A history $h = (a^1, \ldots, a^K)$
is \defn{terminal} if it corresponds to a leaf in the game tree---there is no $K+1$ such that $(a^1, \ldots, a^K, a^{K+1}) \in H$---and \defn{non-terminal} otherwise.

Let $Z(h)$ denote the player whose turn it is to act following a non-terminal history $h$---note that
even though in an ncRIP protocol more than one prover may send a message to the verifier in a round, without loss of generality we can increase the number of rounds
such that only a single prover acts in each round.
Let $A(h)$ denote the set of actions available to the acting player at
a non-terminal history $h$: that is, $A(h) = \{a \mbox{ : } (h, a) \in H\}$.
If $Z(h)$ is Nature, then $A(h)$ is the set of possible coin flips and messages of the verifier following $h$;
otherwise $A(h)$ is the set of possible messages that $Z(h)$ may send to the verifier.
For each terminal history $h$,
the {\em utility} of a player $i$ following $h$, $u_i(h)$,
is the payment $R_i$ computed by the verifier given $x$ and $h$.

As the verifier's coins are private and the verifier exchanges private messages
with each of the provers,
an ncRIP protocol forms an extensive-form
game of imperfect information.

An \defn{information set} $I_i$ of a player $P_i$
is a subset of all possible histories $h$ with $Z(h) = P_i$,
and represents all the information that the player knows when acting in one of the decision nodes in $I_i$.
That is, when a decision node
 in $I_i$ is reached, $P_i$ knows that $I_i$ has been reached but does not know exactly which
node he is at.
The set of actions available to player $i$ at every decision node
in a particular information set is the same, i.e., $A(h) = A(h')$ for all $h, h' \in I_i$.

Let $A(I_i)$ denote the set of available actions at an information set $I_i$.
The set of all information sets of $P_i$ forms a partition of the set $\{h \in H
\mbox{ : } Z(h) = P_i\}$,
and let $\mathcal{I}_i$ to denote this partition, referred to as the {information partition} of $P_i$.
In terms of the protocol, $\mathcal{I}_i$ is in a one-to-one correspondence with
the set of possible message sequences $(m_{i1}, \dots, m_{ij})$ seen by $P_i$, where $j\in \{1, \dots, p(n)\}$
and $P_i$ is acting in round $j$.


A \defn{pure strategy} $s_i$ of a player $P_i$ in an extensive-form game is a function
that assigns an action in $A(I_i)$ to each information set $I_i \in \mathcal{I}_i$.
A \defn{behavioral strategy} $\beta_i$ of $P_i$ is a collection $(\beta_i(I_i))_{I_i \in \mathcal{I}_i}$
of independent probability measures, where $\beta_i(I_i)$ is a probability measure over
the action set $A(I_i)$.
A behavioral strategy $\beta_i$ is \defn{completely mixed} if each $\beta_i(I_i)$ assigns
a  positive probability to every action in $A(I_i)$.

In this paper, the provers are deterministic and thus we only consider pure strategies.
However, the solution
concept introduced in this paper applies to behavioral strategies as well.

A player $i$'s \defn{utility under a strategy profile $s$}, $u_i(s)$, is his expected utility over the
distribution of histories induced
by $s$ and the verifier's randomness.

The provers are computationally unbounded and never ``forget'' anything and thus
the corresponding extensive-form game has \defn{perfect recall}. That is, for any two histories $h$ and $h'$ in the same information set $I_i$ of a player $P_i$,
$h$ and $h'$ pass the same sequence of information sets to player $P_i$.
Furthermore, for any information set in this sequence,
player $P_i$ took the same action in $h$ and $h'$.
%
This holds in any ncRIP protocol since all histories
of prover $P_i$
in the same information set $I_i$ at round $j$
 correspond to the sequence of messages
$(m_{i1}, \dots, m_{ij})$ seen by $P_i$ up to round~$j$.

\subsection{Solution concept for ncRIP}\seclabel{sse}
We want the solution concept for ncRIP to satisfy a strong notion 
of backward induction~\cite{osborne1994course}, a standard criterion
applied to extensive-form games based on the common
knowledge of rationality. Backwards induction refers to the condition 
of being ``sequentially rational'' in an extensive-form game, 
that is, each player must play his best response at each node where he has to move, even if
his rationality implies that such a node will not be reached.
 
%
%
%
%
If an interactive protocol forms an extensive-form game of perfect information, 
it is easy to formalize this condition. A strategy $s$ is
\defn{sequentially rational} or satisfies \defn{backward induction},
if for every player $i$ and every decision node of $i$, conditioned 
on reaching the decision node, $s_i$ is a best response
to $s_{-i}$, that is, $u_i(s_i, s_{-i}) \geq u_i(s_i', s_{-i})$ for any strategy $s_i'$ of prover $i$.
In other words, $s$ induces a best response at every subgame.\footnote{A subgame is a subtree that can be treated as a separate well-defined game. In a perfect-information game, every node starts a new subgame. ``Backward induction'' and ``{subgame-perfect equilibrium}'' are used interchangeably in the literature~\cite{glazer1996virtual}.}

In a game of imperfect information, the decision nodes corresponding to a player's
turn are partitioned into \defn{information sets}, where the player is unable to distinguish
between the possible histories within an information set. 
To reason about sequential rationality we need a probability distribution $u_I$
on each information set $I$, so as to determine the players' expected utility conditioned on reaching $I$
and thus their best response at $I$. 
The probability distribution
$\mu_I$ is referred to as the player's \defn{beliefs} about the potential histories leading to $I$.

Given a strategy profile $s$, beliefs~$u_I$ at \defn{reachable information sets} (reached with non-zero probability under~$s$) are derived from $s$ 
using Bayes' rule; this is a standard derivation used in most solution
concepts for extensive-form games~\cite{osborne1994course}. We sometimes write~$\mu_I^s$ to emphasize that the beliefs depend on~$s$.  

Past work has introduced a variety of methods for defining 
the beliefs $u_I^s$ at \defn{unreachable information sets} $I$ (i.e. information sets reached with 
probability zero under $s$); see e.g.~\cite{kreps1982sequential, 
selten1975reexamination, 
cho1987signaling, 
mclennan1985justifiable}. 
The most well-known is sequential
equilibrium~\cite{kreps1982sequential},
which demands an explicit system of beliefs that satisfies a (somewhat artificial) consistency condition.
Other equilibria, like trembling hand~\cite{selten1975reexamination}, reason implicitly about beliefs at unreachable information sets by assigning a negligible
probability with which the player's hand ``trembles,'' and reaches an otherwise-unreachable information set. 
Further refinements of these  
take the structure and payoffs of the game into account~\cite{cho1987signaling, mclennan1985justifiable, banks1987equilibrium}.  

The treatment of beliefs at unreachable information sets in these solution concepts 
is often focused on ensuring that they can be used to analyze {\em every} extensive-form game. From a mechanism-design
perspective, our focus is different---we want to design mechanisms in such a way that they admit much stronger equilibrium requirements, even if such an equilibrium cannot be used to analyze every game.

At a high-level, 
we want the players' beliefs to be irrelevant in determining 
their best response at unreachable information sets. 
We call this notion \defn{strong sequential rationality}.
A strategy profile $s$ is \defn{strongly sequentially rational} if for every
information set $I$, conditioned on reaching $I$, $s_i$
is a best response to $s_{-i}$ with respect to $\mu_I^s$, where
\begin{itemize}
\item $\mu_I^s$
is derived using Bayes's if $I$ is reachable under $s$, and 
\item $\mu_I^s$
is \emph{any} arbitrary probability distribution if $I$ is unreachable under $s$.
\end{itemize} 
In~\secref{strongse}, we show that 
this requirement 
is equivalent to saying that, at an unreachable information set $I$, $s_i$
must be a best response to $s_{-i}$ conditioned on reaching each history $h \in I$.
In other words, at an unreachable information set $I$, each player must have a \emph{single} action that is the best response to every possible history in $I$.  
We say a strategy profile is a \defn{strong sequential equilibrium} (SSE) 
if it satisfies
strong sequential rationality.

We refine our solution concept further 
to eliminate strategies that are weakly dominated within ``subgames'' of the entire game. 
This is crucial to deal with equilibrium selection, in particular, because
the players' cannot unilaterally deviate out of a suboptimal equilibria.
We say an SSE $s$ \defn{weakly dominates} another SSE $s'$ 
if, for any player $i$, $u_i(s) \geq u_i(s')$. A strategy $s$ is \defn{weakly dominant}
if it dominates all SSEs. Next we eliminate SSEs
that are weakly dominated in subgames of the entire game.
We use the generalized notion of subgames, called \defn{subforms}, defined by Kreps and Wilson~\cite{kreps1982sequential}
for extensive-form games with imperfect information.

To review the definition of subforms, we need further notation.
Let $H$ be the set of histories of the game. Recall that
a history is a sequence $(a^1, \ldots, a^K)$ of actions taken by the players.
For histories $h, h' \in H$, we say $h$ has $h'$ as a \defn{prefix}
if there exists some sequence of actions $b^1, \ldots, b^L$ (possibly empty)
such that $h = (h', b^1, \ldots, b^L)$.
For a history $h \in H$, let $I(h)$ be the unique
information set containing $h$. 

For an information set $I$, let $H_I$ be the set of all histories
following $I$, that is, $H_I$ is the set of all histories $h \in H$
such that $h$ has a prefix in $I$. 
We say that
$H_I$ 
is a \defn{subform \em rooted at $I$} if for every information set $I'$ such that $I' \cap H_I \neq \emptyset$, 
it holds that $I' \subseteq H_I$.
%
%
Roughly speaking, a subform $H_I$ ``completely contains''
all histories of the information sets following $I$, so there is no information asymmetry between
the players acting within $H_I$. 
%

Thus, given a strategy profile, the subform $H_I$ together with the probability distribution $\mu_I^s$ on $I$, 
can be treated as a well-defined game.

We say an SSE $s$ \defn{weakly dominates} SSE $s'$ \defn{on a subform} $H_I$ if,  
for any player $j$ acting in $H_I$,
the expected utility of $j$ under $s_I$ in the game $(H_I, \mu_I^s)$ is greater
than or equal to their utility under $s_I'$ in the game $(H_I, \mu_I^{s'})$.

We eliminate weakly dominated strategies by imposing this dominance condition in a backward-induction-compatible way on the subforms
as follows.


%
%
\begin{definition}[Dominant Strong Sequential Equilibrium]
\deflabel{max-sse}
A strategy profile $s$ is a dominant strong sequential equilibrium if
$s$ is an SSE and
  \begin{itemize}
\item for every subform $H_I$ of height $1$: $s$ weakly dominates $s$' on $H_I$ for any SSE $s'$ 
\item for every subform $H_I$ subgame of height $>1$:
%
$s$ weakly dominates $s'$ on $H_I$ for any SSE $s'$ that is a dominant SSE in all subforms of height at most $h-1$. 
\end{itemize}
\end{definition}

We are ready to define non-cooperative rational interactive proofs. 
\begin{definition}[Non-Cooperative Rational Interactive Proof]
\deflabel{mripnc} Fix an arbitrary string $x$ and language $L$.
An interactive protocol $(V, \vec{P})$ is a {\em non-cooperative
rational interactive proof} (ncRIP) protocol for $L$ if
there exists a strategy profile $s$ of the provers 
that is a dominant SSE in the resulting extensive-form game, and
under any dominant SSE,
the answer bit $c$
output by the verifier
is correct (i.e., $c=1$ iff $x\in L$) with probability 1, where the probability is taken over the verifier's randomness.
\end{definition}

\subsection{Utility Gap in ncRIP Protocols}\label{sec:gap-model}
In game theory, players are assumed to be perfectly rational and  ``sensitive'' to
arbitrarily small utility losses.
%
In reality, some provers may not care about small losses. Such provers may not
have sufficient incentive to reach a dominant SSE, and could end up leading the verifier to the wrong answer.
To design ncRIP protocols that are robust against such ``insensitive'' provers,
we define the notion of \defn{utility gap}.

%
Informally, a utility gap of $u$ means that if a strategy profile $s$
leads the verifier to the wrong answer, there must exist a subform, such that
some provers must lose at least a $1/u$ amount in their final individual payments (compared to their optimal strategy in that subform).
As a consequence, these provers will not deviate to $s$, {as long as} they care about $1/u$ payment losses.
We formalize this notion below.  (We say a subform $H_I$ is reachable under $s$ if the information set $I$
is reached under $s$ with non-zero probability.)

\begin{definition}[Utility Gap]\deflabel{rewardgap}
Let~$(V, \vec{P})$ be an ncRIP protocol for a language~$L$ and~$s^*$ be a dominant SSE of the resulting game.
The protocol~$(V, \vec{P})$ has an {\em $\alpha(n)$-utility gap} or $\alpha(n)$-gap, if
for any strategy profile $s'$ under which the answer bit~$c'$ is wrong,
there exists a subform~$H_{I}$ reachable under $s'$, and a prover~$P_j$
acting in~$H_I$ who has deviated from~$s^*$ such that
\[u_j (x, (s_{-I}', s_{I}^*), (V, \vec{P})) -
u_j (x, (s_{-I}', s_{I}'), (V,\vec{P})) > 1/\alpha(n),\]
where $s_{-I}'$ denotes the strategy profile $s'$ outside subform $H_I$,
that is, $s_{-I}'=s' \setminus s_{I}'$.
\end{definition}

%

The class of languages that have an ncRIP protocol with \defn{constant},
\defn{polynomial} and \defn{exponential} utility gap, are denoted by $\cncRIP$, $\pncRIP$, and $\encRIP$ respectively.\footnote{These
classes are formally defined by taking the union over languages with
$\alpha(n)$ utility gap, for every $\alpha(n)$ that is constant, polynomial and exponential in $n$ respectively.}
Note that $\alpha(n)$ gap corresponds to a payment loss of $1/\alpha(n)$, so an exponential utility gap is the weakest guarantee.



\section{Lower Bounds: ncRIP Protocols with Utility Gap}\seclabel{lower}

In this section, we give an $O(1)$-utility gap ncRIP protocol for the class $\sf{NEXP}$
and use it to give an $O(\alpha(n))$-utility gap ncRIP protocol for the class $\sf{P^{NEXP[\alpha(n)]}}$.
Setting $\alpha(n)$ to be a constant or polynomial in $n$ gives us $\sf{P^{NEXP[O(1)]}} \subseteq \cncRIP$
and $\sf{P^{NEXP}} \subseteq \pncRIP$
respectively.



\paragraph*{A constant-gap ncRIP protocol for {\boldmath $\sf{NEXP}$}}
The ncRIP protocol for any language in~$\sf{NEXP}$ is in~\figref{nexp-ncrip-protocol}.
The protocol uses the 2-prover 
1-round MIP for $\sf{NEXP}$~\cite{feige1992two} as a blackbox.\footnote{It is also possible
to give a scoring-rule based ncRIP protocol for $\sf{NEXP}$, similar to MRIP~\cite{ChenMcSi16}. However,
such a protocol has an exponential utility gap.}
The protocol in~\figref{nexp-ncrip-protocol} essentially forces the non-cooperative provers to coordinate by
giving them identical payments.  As a result, it is almost identical to the MRIP protocol for $\sf{NEXP}$~\cite{ChenMcSi16}.

While the payment scheme is simple, in the analysis we have to open up the black-box MIP. In particular,
if $P_1$ sends $c=0$ in round~\ref{step:answer}, all the information sets 
of $P_1$ and $P_2$ in round~\ref{simple-3} become unreachable. To show that
an SSE exists, we show that the provers have
a best response at these unreachable sets, which is argued based on the 
messages exchanged in the MIP protocol.

\begin{lemma}\lemlabel{nexp-const}
\lemlabel{nexp-mip}
Any language $L \in \sf{NEXP}$ has a 2-prover 3-round $6/5$-gap ncRIP protocol.
\end{lemma}

\begin{proof} 
The ncRIP protocol for any language $L \in \sf{NEXP}$ is given in~\figref{nexp-ncrip-protocol}.

We show that there exists a strategy profile $s = (s_1,s_2)$ of provers $P_1$ and $P_2$ respectively that
is a dominant SSE of the game tree corresponding to the protocol $(V, P_1, P_2)$ and under 
any dominant SSE, the answer bit $c=1$ if and only if $x \in L$.

In the protocol, if $c=0$, no player acts. If $c=1$, the
verifier executes the 1-round blackbox MIP protocol with $P_1$ and $P_2$.
To exhibit a strategy that is a best response for $P_1$ and $P_2$ on their information sets at step~\ref{simple-3},
we look at the messages the verifier sends to each prover in the classic $\sf{MIP}$ protocol.
In the $\sf{MIP}$ protocol, the verifier sends $P_1$ a tuple of message pairs $\vec{m}_1 = ((q_1, x_1),\ldots, (q_m,x_m))$
where $m$ is a polynomial in $n$ and $V$ sends $P_2$ a tuple of random messages $\vec{m}_2 = (y_1, \ldots, y_m)$. $P_1$
sends back a polynomial $P(t)$ and $P_2$ sends back the value of the polynomial $P(t)$ for $t$ satisfying $q_j + t x_j = y_j$.
The verifier rejects if their answers are inconsistent. 

To analyze the SSE strategy, without loss of generality, suppose $P_1$ moves last in the MIP protocol.
Any information set $I_1$ of $P_1$ at step~\ref{simple-3} is 
characterized by the message $\vec{m}_1$ he receives. The decision nodes in $I_1$ correspond to each
possible message $\vec{m}_2$ that $P_2$ could have received. 

Because the $V$ gives the largest payment when the MIP protocol accepts, given $P_2$'s strategy, if any information set $I_1$
of $P_1$ is reached under $s$ then $P_1$'s best response at $I_1$ is to maximize the acceptance-probability 
of the MIP protocol given his beliefs on $I_1$. Similarly, given $P_2$'s strategy, if any information set $I_1$
of $P_1$ is unreachable under $s$ then, $P_1$'s best response at $I_1$ for every decision node in $I_1$
is the following: given $\vec{m}_1 = ((q_1, x_1),\ldots, (q_m,x_m))$, respond with a polynomial $P(t)$ such that
$P(t)$'s value at all $t$ coincides with $P_2$'s reply on all $y_j$ where $q_j + t x_j = y_j$.

Given $P_1$'s strategy of committing
to a polynomial $P(t)$ that matches $P_2$ on all values of $t$, $P_2$' best response
at any information set $I_2$ (reachable or unreachable under $s$) at step~\ref{simple-3} at every decision node in $I_2$ 
is to answer the tuple of queries $(y_1, \ldots, y_m)$ so as to maximize the acceptance probability of the MIP protocol. 
The verifier's move at step~\ref{simple-3} is the root of a non-trivial subform.
Conditioned on step~\ref{simple-3} being reached, any dominant SSE at this subform
corresponds to a strategy profile $s$ that is an SSE, which when restricted to this subform,
maximizes the acceptance probaility of the MIP protocol. Under any such dominant SSE, we show that
$P_1$'s best response at step~\ref{step:answer} is to send the correct answer bit.

Suppose $x \in L$. If $P_1$ sends $c=0$, then $R_1=1/2$ with probability $1$. 
On the other hand, if $P_1$ sends
$c=1$, by the soundness condition of the MIP protocol, the acceptance probability
is $1$, leading to $R_1 =1$. Thus for $x \in L$, $s$ is a dominant SSE iff $P_1$ sends $c=1$.

Suppose $x \notin L$. If $P_1$ reports $c=0$, then $R_1=1/2$ with probability $1$. 
On the other hand if $P_1$ reports $c=1$, then by the soundness condition of the
MIP protocol, the maximum
acceptance probability is $1/3$
leading to $R_1=1$.  The protocol rejects with probability at least
$2/3$ leading to $R_1=-1$.  Thus, $P_1$'s expected payment for
misreporting the answer bit is at most $R_1=-1/3$.  Thus for $x
\notin L$, $s$ is a dominant SSE iff $P_1$ sends $c=0$.

Thus, under $s$ which is a dominant SSE, $c=1$ if and only if $x \in L$.

Furthermore, the payment incurred by the provers when the answer bit sent
in the first round is incorrect is at least $5/6$ for both provers and thus
the protocol has constant utility gap.
\end{proof}

\vspace{-10pt}
\begin{figure}[phtb]
\centering
\fbox{
\begin{minipage}{0.96\textwidth}
{\normalsize
\noindent{}For any input $x$ and language $L \in \sf{NEXP}$, the protocol $(V, P_1, P_2)$ for $L$ is:
\begin{enumerate}
\item\label{step:answer} $P_1$ sends a bit $c$ to $V$.  $V$ outputs $c$ at the end of the protocol.

\item\label{simple-2} If $c=0$, then the protocol ends and the payments are $R_1 =R_2 = 1/2$.

\item\label{simple-3}
Otherwise, $V$ runs the classic 2-prover 1-round MIP protocol for $\sf{NEXP}$~\cite{feige1992two} with $P_1$ and $P_2$ to prove if $x \in L$.
If the MIP protocol accepts then $R_1 = 1$, $R_2=1$; else, $R_1 = -1$, $R_2=-1$.
\end{enumerate}
}
\end{minipage}
}
\caption{A simple $O(1)$-utility gap ncRIP protocol for $\sf{NEXP}$.}
\figlabel{nexp-ncrip-protocol}
\end{figure}
\vspace{-12pt}
\paragraph*{An {\boldmath $O(\alpha(n))$}-gap ncRIP protocol for {\boldmath$\sf{P^{NEXP[\alpha(n)]}}$}}
Using the above $\sf{NEXP}$ protocol as a subroutine,  we give an ncRIP protocol with~$O(\alpha(n))$-utility gap
for the class $\sf{P^{NEXP[\alpha(n)]}}$.  This protocol works for any function $\alpha(n)$
which~(1) is a positive integer for all $n$,~(2) is upper-bounded by a polynomial in~$n$,~and~(3) is
polynomial-time computable.\footnote{For~\thmref{constchar}
and~\thmref{polychar},~$\alpha(n)$ need only be a 
constant or polynomial in $n$. However,~\lemref{notc-lower} holds for all $\alpha(n)$'s that are
polynomial-time computable (given $1^n$) and polynomially bounded, such as $\log n$, $\sqrt{n}$, etc.} 
%

The ncRIP protocol for any $L\in \sf{P^{NEXP[\alpha(n)]}}$ is in \figref{constpolygap}.  
It is fairly intuitive---$V$ simulates the polynomial-time machine directly, and uses the
ncRIP protocol for~$\sf{NEXP}$ for the oracle queries.  

\begin{figure}[phtb]

\centering \fbox{ \begin{minipage}{0.96\textwidth} {\normalsize 
\noindent{}For any input $x$ of length $n$, the protocol $(V, \vec{P})$ works as follows.  

\begin{enumerate}

\item\label{step:answer-bits} $P_1$ sends $(c, c_1,\ldots,
c_{\alpha(n)}) \in \{0,1\}^{\alpha(n)+1}$ to $V$. $V$ outputs $c$ at the end of the protocol.  

\item\label{step:output-test} $V$ simulates $M$ on $x$ using the
bits $c_1,\ldots,c_{\alpha(n)}$ as answers to $\sf{NEXP}$ queries $\phi_1, \ldots,
\phi_{\alpha(n)}$  generated by $M$ respectively. If $M$ accepts and $c = 0$
or $M$ rejects and $c=1$, then the protocol ends and $R_1=-1, R_2 =R_3=0$.

\item\label{step: q-index} $V$ picks a random index $i'$ from $\{1, \ldots,
\alpha(n)\}$ and sends $(i', \phi_{i'})$ to $P_2$ and $P_3$.

\item\label{step: oracle} $V$ runs the 2-prover 3-round $O(1)$-gap ncRIP
    protocol for $\sf{NEXP}$~(\figref{nexp-ncrip-protocol}) with $P_2$ and $P_3$
on~$\phi_i$.~$P_2$ and~$P_3$ get payments~$R_2$ and~$R_3$ based on the
protocol.  Let $c_{i'}^*$ be the answer bit in the
$\sf{NEXP}$ protocol. If $c_{i'}^* \neq c_{i'}$, then $R_1=0$;
otherwise $R_1 =1$.
\end{enumerate}
}
\end{minipage}
}
\caption{An $O(\alpha(n))$-utility gap ncRIP protocol for $\sf{P^{NEXP[\alpha(n)]}}$.}
\figlabel{constpolygap}
\end{figure}

We first argue the correctness of this protocol at a high-level and then present the formal proof. 
Under any strategy of $P_1$, the resulting $\sf{NEXP}$ queries in the protocol in~\figref{constpolygap} 
are the roots of non-trivial subforms.  Which of these subforms
are reachable under a strategy profile $s$ is determined solely by the strategy of $P_1$.  However,
because weak dominance is imposed on all subforms in a bottom-up fashion, 
$P_2$ and $P_3$ must
play their optimal strategy in these subforms regardless of their reachability---and therefore, they must play optimally for any strategy of $P_1$. (This is one example of why ruling out weakly-dominated strategies in subforms in the definition of dominant SSEs is crucial to arguing correctness.)  
From the correctness of the $\sf{NEXP}$ protocol in
\figref{nexp-ncrip-protocol}, we know that the optimal strategy of $P_2$ and $P_3$ is to compute the $\sf{NEXP}$ queries
correctly.  Given that the best response of $P_2$ and $P_3$ is to solve the $\sf{NEXP}$ queries correctly, 
and given that $V$ randomly verifies $1$ out of $\alpha(n)$ queries, $P_1$ must commit
to correct answer bits in the first round, or risk losing a $O(1/\alpha(n))$ amount from his expected payment.

%
%
If $P_1$ gives the correct answer bits in step 1, but $P_2$ or $P_3$ deviate within a subform corresponding
to an $\sf{NEXP}$ query $\phi_q$, then with probability $1/\alpha(n)$, $V$ simulates the protocol in \figref{constpolygap}
on $\phi_q$, in which case they lose a constant amount of their expected payment.



\begin{lemma}\label{lem:notc-lower}
Any language~$L \in \sf{P^{NEXP[\alpha(n)]}}$ has a 3-prover 5-round ncRIP protocol
that has a utility gap of $6/(5\alpha(n))$.
\end{lemma}

\begin{proof} 
Consider any language $L \in \sf{P^{NEXP[\alpha(n)]}}$.  Let $M$ be a
polynomial-time Turing machine deciding $L$, with access to an oracle $O$ for
an $\sf{NEXP}$ language.  

The ncRIP protocol for $L$ is given in \figref{constpolygap}.   

Let $s_1, s_2, s_3$ denote the strategy used by $P_1$, $P_2$ and $P_3$ for the
protocol in \figref{constpolygap}, and $s=(s_1, s_2, s_3)$.  First, note that
regardless of $s_2$ and $s_3$, $P_1$'s best response at
step~\ref{step:answer-bits} is to send the bits $c, c_1, \ldots, c_\alpha(n)$ such
that the verification in step~\ref{step:output-test} goes through.  In
particular, if $s_1$ is such that the output of $M$ on input $x$, using
$c_1,\ldots, c_{\alpha(n)}$ as answers to $\sf{NEXP}$ queries $\phi_1, \ldots,
\phi_{\alpha(n)}$ is consistent with $c$, then $P_1$ gets $R_1\geq 0$.
Meanwhile, if the verification in step~\ref{step:output-test} fails then
$R=-1$.  Thus, under any SSE $s$, the answer bits $c_1$, \ldots, $c_{\alpha(n)}$
sent by $P_1$ must be consistent with the computation of $M$ on $x$ and the
final the answer bit $c$, regardless of $s_2$ and $s_3$. 

We now argue using backward induction.  Each random index $i'$ chosen by $V$ in
step~\ref{step: q-index} together with $\phi_{i'}$ starts a subform.  In
particular, since $P_2$ and $P_3$ both know $(i', \phi_{i'})$, all their
information sets starting from step~\ref{step: oracle} are completely disjoint
from information sets reached under a different index and $\sf{NEXP}$ query.
By \lemref{nexp-mip}, there exists a dominant SSE $s$ on each such
subform simulating an $\sf{NEXP}$ query, and under any dominant
SSE, $s_2$ and $s_3$ are such that $c_{i'}^*$ is the correct answer to the
$\sf{NEXP}$ query. 

Moving up the tree, the next subform is induced by $V$'s nature move at
step~\ref{step: q-index} assigning a probability to each subsequent 
subform. Since under any dominant SSE, the expected payments of $P_2$
and $P_3$ (conditioned on reaching these subforms) are maximized, the
overall expected payments under $V$'s nature move at step~\ref{step: q-index}
is also maximized.

We move up a further level in the tree to the root.  We show that $P_1$'s best
response at step~\ref{step:answer-bits} is to send the correct answer bits,
given that under any dominant SSE $s$:
\begin{itemize}[topsep=0pt,noitemsep] \item $P_2$ and $P_3$ answer each
$\sf{NEXP}$ query $\phi_{i'}$ determined by $s_1$ and index $i'$ correctly, and
\item the verification in step~\ref{step:output-test} goes through (i.e. $P$
does not set $R_1 = -1$) under $s_1$.  \end{itemize}

Suppose $s_1$ is such that there exists an $\sf{NEXP}$ query where $P_1$ lies.
Let $k$ be the first $\sf{NEXP}$ query index such that $c_k$ is not the correct
answer to query $\phi_k$, where $1 \le k \le \alpha(n)$. In particular, the
instance $\phi_k$ is evaluated correctly (by running $M$ on $x$ using the
correct answers to previous queries, $c_1, \ldots, c_{k-1}$) but the answer
$c_k$ is not evaluated correctly based on $\phi_k$.  Then with probability $1/\alpha(n)$, 
$V$ picks $k$ in step~\ref{step: q-index} and crosschecks the $c_k$
with $c_{i'}^*$, in which case the verification fails and $R_1 =0$.  Thus, $P_1$'s
expected payment is at most $1-1/\alpha(n)$.  If $P_1$ answers all $\sf{NEXP}$
queries correctly, since the verification in step~\ref{step:output-test} goes
through, $P_1$ gets $R_1=1$ with probability $1$. Thus, $c, c_1, \ldots,
c_{\alpha(n)}$ are correct under any dominant SSE $s$, and $c=0$ if and
only if $x \in L$. 

Now, we show that protocol $(V, \vec{P})$ has $O(\alpha(n))$ utility gap. Let $s^*$
be a dominant SSE of the game resulting from $(V, \vec{P})$.
Suppose $s'$ is such that the answer bit $c'$ under $s'$ is incorrect. We go ``bottom-up''
in the game tree and exhibit a subform $H_I$ (reachable under $s'$) such that 
some prover acting in that subform loses $O(1/\alpha(n))$ compared to the strategy 
where $s^*_{I}$ is played on $H_I$, keeping the rest of the strategy fixed.

First, consider all the $\sf{NEXP}$ queries at step~\ref{step: oracle} that
start subforms.  Suppose there exists a query $\phi_k$ committed under
$s_1'$, for $1 \le k \le \alpha(n)$, such that $c_k*$ is the wrong answer to
$\phi_k$. By~\lemref{nexp-const}, both $P_2$ and $P_3$ lose a constant amount
($5/6$ in particular) from their expected payment (conditioned on reaching this
subform) compared to the dominant SSE strategy profile $s_{\phi_k}^*$
which reports the correct answer to $\phi_k$.  Since $V$ chooses  $\phi_k$ with
probability $1/\alpha(n)$, $P_2$ and $P_3$ can gain $O(1/\alpha(n))$ in their overall
expected payment by deviating to strategy profile $s_{\phi_k}$, at the
subform corresponding to $(k,\phi_k)$ keeping $s_{-\phi_k}'$ fixed.
Specifically, \[\mu_i \left(x, r, (s_{-\phi_k}', s_{\phi_k}^*), (V,
\vec{P})\right) - \mu_i \left(x, r, (s_{-\phi_k}', s_{\phi_k}'), (V,
\vec{P})\right) > \frac{1}{\alpha(n)} \left(\frac{5}{6}\right),~~\mbox{for}~~i \in \{2,3\}.\]

Finally, suppose $P_2$ and $P_3$ answer all $\sf{NEXP}$ queries (reachable
under $s'$) correctly.  Then, $P_1$ loses at least $1/\alpha(n)$
at the subform at the root---the entire game.  Since the answer bit $c'$
under $s'$ is incorrect, either step~\ref{step:output-test} fails or $P_1$ lies on
some $\sf{NEXP}$ query. In the first case, $P_1$ gets $-1$ with probability $1$
compared to an expected payment of $1$ under $s^*$. In the second case, $P_1$ gets caught
in step~\ref{step: oracle} with probability $1/\alpha(n)$, and gets an expected
payment of at most $1-1/\alpha(n)$, losing at least $1/\alpha(n)$ compared to $s^*$.

Thus, the protocol $(V, \vec{P})$ is an ncRIP protocol for
$\sf{P^{NEXP[O(\alpha(n)])}}$ and has $O(\alpha(n))$ utility gap.   
\end{proof}

\paragraph*{Exponential utility gap} 
We show how to simulate a general MRIP protocol $(V, \vec{P})$ with $p(n)$ provers and $k(n)$ rounds for a language $L$
using a 2-prover 3-round ncRIP protocol  $(V', {P_1'},P_2')$ with exponential-utility gap. (The protocol $(V', {P_1'},P_2')$ is in~\figref{mrip-to-ncrip}.)

Essentially, $V'$ gives all the randomness of $V$ to $P_1'$ and asks
for the entire transcript and uses $P_2'$ to commit to a single prover's message, and cross-checks their answers. 
However, we don't want $P_1'$ who has access to all the randomness to
dictate what information sets of $P_2'$ are reachable. 
Because the ncRIP protocol only needs an exponential utility gap, $V'$ asks one prover
a totally random question (independent of $P_1'$), and with
exponentially small probability this random message is exactly the message $V'$ intended
to check. This protocol shows why exponential gap guarantees do not lead to meaningful protocols---a verifier that
asks random questions can still extract honest behavior from rational provers through the exponentially small changes in expected payments.

\begin{lemma}\lemlabel{mrip-ncrip}
Any MRIP protocol can be simulated using a $2$-prover $3$-round ncRIP protocol with $O(1/2^{n^k})$-utility gap,
for some constant $k$, where $n$ is the length of the input.
\end{lemma}

\begin{proof} 

Without loss of generality, let each message in the protocol be of length $\ell(n)$ for any input of length $n$,
where $\ell(n)$ is a polynomial in $n$.
We shift and rescale the payment function of~$V$, so that 
the payment is always in $[0, 1]$, and 
the expected payment is strictly greater than $0$ under the provers' best strategy profile.

We simulate $(V,\vec{P}')$ using an ncRIP protocol $(V', (P_1', P_2'))$, given in~\figref{mrip-to-ncrip}.

\begin{figure}[tbhp] \centering
\fbox{
\begin{minipage}{0.96\textwidth}
{\normalsize

\noindent
Given an input $x$ of length $n$, and an MRIP protocol $(V, \vec{P})$, the ncRIP protocol $(V', \vec{P}')$ is:

\begin{enumerate}[leftmargin=15pt, nolistsep, noitemsep]
\item\label{firstround} 
$P_1'$ sends the round 1 messages  $m_{11}, \dots, m_{p(n)1}$ of $(V, \vec{P})$ to $V'$.
$V'$ outputs $c$, the first bit of $m_{11}$, at the end of the protocol.

\item\label{newrandomness} $V'$ selects a random prover index $i \in \{1,\ldots, p(n)\}$ and a random round $j\in \{1, \ldots, k(n)\}$.
Then, $V'$ generates a random string $\tilde{m}_{ij}$ of length $(j-1)\ell(n)$.

\item\label{p2-response} $V'$ sends $(i, j, \tilde{m}_{ij})$ to $P_2'$. $P_2'$ simulates $P_i$ on round $j$, and
sends the message $m'_{ij}$ to $V'$.

\item\label{originalrandomness}
$V'$ generates all the randomness $r$ used by $V$ and sends it to $P_1'$.

\item\label{p1-response}
$P_1'$ uses $r$ to simulate the protocol $(V, \vec{P})$, and
sends the resulting transcript $\vec{m}$ to $V'$.

\item If~$\tilde{m}_{ij} \neq (m_{i1},\ldots, m_{i(j-1)})$, where~$m_{ij}$ denotes prover~$P_i$'s message in round~$j$ according to~$\vec{m}$
sent by $P_1$',
then the protocol ends and~$R_1'=R_2'=0$.

\item\label{consistency} Otherwise, if $m_{ij} \neq m'_{ij}$, then $R_1' =R_2'=-1$.

\item\label{correctness} Else, $V'$ computes the payment $R$ in $(V, \vec{P})$ using $x$, $r$ and $\vec{m}$,
and sets $R_1' =0$, $R_2'=R$.
\end{enumerate}
}
\end{minipage}
}
\caption{Simulating any MRIP using an ncRIP protocol with exponential utility gap.}\figlabel{mrip-to-ncrip} 
\end{figure}

Let $s_1'$ and $s_2'$ denote the strategy of the provers $P_1'$ and $P_2'$
respectively and $s'=(s_1',s_2')$. Since $P_2'$ is queried only once and about
a single message in Step~\ref{p2-response}, any strategy $s_2'$ of $P_2'$  de
facto commits to a strategy profile for the provers in $(V, \vec{P})$.

We analyze the game tree of the protocol $(V', \vec{P}')$ bottom-up. 

The last move is by $P_1'$ sending the entire transcript $\vec{m}$ at
step~\ref{p1-response}.  Any information set $I_1'$ of $P_1'$ is characterized
by the randomness $r$ received by $P_1'$~in~step~\ref{originalrandomness} and
all information sets are reachable under any $s'$. The decision nodes in $I_2'$
correspond to different strings $\tilde{m}_{ij}$ that $P_2'$ could have been
asked in~step~\ref{newrandomness}.  Given $s_2'$, the best response of $P_1'$
at any information set $I_1'$, for any beliefs at $I_1'$, is to match the
transcript committed by $P_2'$ and make the verification
in~step~\ref{consistency} go through.  Suppose there exists a prover index $i$
and round $j$ such that the message $m_{ij}$ in $\vec{m}$ that is inconsistent
with the corresponding message $m_{ij}'$ committed under $s_2'$.  With
probability $\frac{1}{2^{(j-1)\ell(n)}}$, the random string $\tilde{m}_{ij}$
generated by $V'$ in Step~\ref{newrandomness} is equal to $(m_{i1},\dots,
m_{i(j-1)})$, otherwise the protocol ends with $R_1'=0$.  With probability at
least $\frac{1}{p(n)k(n)}$, $V'$ chooses $(i,j)$ in~step~\ref{newrandomness},
and queries~$P_2'$ for $m_{ij}'$ and $R_1'=-1$. If $(i,j)$ is not chosen then
$R_1'=0$. Thus, $P_1'$ expected payment at $I_1'$ is at most 
\[\sum_{i\leq p(n), 1\leq j\leq k(n)} \frac{1}{2^{(j-1)\ell(n)}} \cdot \frac{1}{p(n)k(n)}
\cdot \left( \mathbb{I}_{m_{ij}\neq m_{ij}'}\cdot  (-1) +  \mathbb{I}_{m_{ij}=
m_{ij}'} \cdot  0\right) < 0.\] 
On the other hand, matching $s_2'$ on all
messages leads to an expected payment of $0$ at $I_1'$ for $P_1'$.

Given that $P_1'$ best response is to make the  verifier
in~step~\ref{consistency} go through for every randomness $r$, we analyze
$P_2'$ move at step~\ref{p2-response}. Any information set $I_2'$ of $P_2'$ is
characterized by the random string $\tilde{m}_{ij}$ received by
$P_2'$~in~step~\ref{newrandomness} and all information sets are reachable under
any $s'$.  The decision nodes in $I_1'$ correspond to different random strings
$r$ that $P_1'$ could have been asked in~step~\ref{newrandomness}.  The best
response of $P_2'$ at any information set $I_1'$, for any beliefs at $I_1'$, is
to commit to the correct strategy profile $s$ of the provers $\vec{P}$. Suppose
$P_2'$ commits to a strategy profile $s'$ such that the answer bit under $s'$
is wrong. With probability $\frac{1}{2^{(j-1)\ell(n)}}$, the random string
$\tilde{m}_{ij}$ generated by $V'$ in Step~\ref{newrandomness} matches
$(m_{i1},\dots, m_{i(j-1)})$, otherwise the protocol ends with $R_2'=0$.  If it
matches, then $P_2'$ expected payment is determined by the expected payment
that $\tilde s$ gets in $(V, \vec{P})$ given $x$ and randomness $r$, which is
strictly less than the expected payment under the strategy profile $s$ which
commits to the correct answer bit (by correctness of the original MRIP
protocol). That is,
\[\sum_{1\leq j\leq k(n)} \frac{1}{k(n)} \cdot \frac{1}{2^{(j-1)\ell(n)}} \cdot
u_{(V, \vec{P})}(x,\tilde{s}) <\sum_{1\leq j\leq k(n)} \frac{1}{k(n)} \cdot
\frac{1}{2^{(j-1)\ell(n)}} \cdot u_{(V, \vec{P})}(x,{s}).  \]
Thus, given that $s_1'$ matches $s_2'$ for every randomness $r$, the best
response by $P_2'$ is to commit to a strategy profile $s_2'=s$ that maximizes
the total expected payment of the original protocol $(V,\vec{P})$ and thus has
the correct answer bit.

There are no non-trivial subform in the game. Any weakly-dominant SSE is a
dominant SSE, under which both $P_1'$ and $P_2'$ maximize their
expected payments---$P_1'$ matches $P_2'$ on all messages and $P_2'$ commits to
the correct strategy profile $s$. Thus, the protocol $(V, \vec{P})$ is correct.
\end{proof}

\section{Upper Bounds: ncRIP Protocols with Utility Gap}\seclabel{upper}

In this section, we prove matching upper bounds on the classes of ncRIP protocols with constant
and polynomial utility gaps.  In particular, we show that any language in $\cncRIP$ (or $\pncRIP$)
can be decided by a polynomial-time Turing machine with a constant (resp. polynomial) number of queries 
to an $\sf{NEXP}$ oracle.

To simulate an ncRIP protocol, we need to find a strategy profile ``close enough'' to the dominant SSE
so that the answer bit is still correct, i.e. a strategy profile that satisfies the utility-gap guarantee. 
We formalize this restatement of~\defref{rewardgap} below. 
\begin{observation}\label{gaprestate}
Given input $x$ and an ncRIP protocol $(V, \vec{P})$ with $\alpha(n)$-utility gap,
let $s$ be a strategy profile such that for all reachable subforms $H_I$ 
and all provers $P_j$ acting in $H_I$,
\[u_j (x, r,(V, \vec{P}), (s_{-I}, s_{I}^*)) -  u_j (x, r,(V,\vec{P}), (s_{-I}, s_{I})) < \frac{1}{\alpha(n)},\]
where $s^*$ is a dominant SSE. Then, the answer bit $c$ under $s$
must be correct.
\end{observation}

There are several challenges involved in finding a strategy profile satisfying Observation~\ref{gaprestate}.

First, the size of the game tree of any ncRIP protocol---small gap notwithstanding---can be exponential in $n$.
Even if the polynomial-time machine considers a single strategy profile $s$ at a time,
since $V$ can flip polynomially many coins, the part of the tree
``in play''---the number of decision nodes reached with positive probability under $s$---can be exponential in $n$.

The second (and related) challenge is that of verifying whether a strategy profile is a dominant SSE. 
While the $\sf{NEXP}$ oracle can guess and verify an SSE, it cannot 
directly help with dominant SSEs. The polynomial-time machine must check using backward
induction if an SSE
is dominant on all its reachable subforms, which can again be exponential in $n$.

Finally, the polynomial-time machine needs to search through
the exponentially large strategy-profile space in an efficient way to find one which leads to the correct answer. 


In the remainder of the section we address these challenges.  In \lemref{pruning} we show that we can prune the game tree, resolving the first two challenges.  Then in Lemmas \ref{intervalsse} and \ref{lem:maxsse}, we show how to efficiently search through the strategy-profile space.

\paragraph*{Pruning Nature moves in ncRIP protocols}
We now give our main technical lemma for the upper bound, which shows that we can limit ourselves to examining protocols with bounded game trees without loss of generality.

Recall that a verifier's coin flips in an ncRIP protocol represent {Nature moves} in the resulting game.
The problem is that a polynomial-time verifier 
can result in Nature moves that impose nonzero probabilities over exponentially many outcomes.

We prune the Nature moves of a verifier so that a polynomial-time
Turing machine simulating an $\alpha(n)$-utility-gap protocol can traverse the game tree
reachable under a given $s$.  This pruning operation takes exponential time (linear in the size
of the game tree), and can be performed by the $\sf{NEXP}$ oracle.


\begin{lemma}[{\bf Pruning Lemma}]\lemlabel{pruning}
Let $L \in \ancRIP$ and let $(V, \vec{P})$ be an ncRIP protocol
for $L$ with $\alpha(n)$ utility gap and $p(n)$ provers.
Given an input $x$ and a strategy $s$,
the protocol $(V, \vec{P})$ can be transformed in exponential time to a
new protocol $(V', \vec{P})$,
where
\begin{itemize}
    \item the probability distribution on the outcomes imposed by the Nature moves of $V'$ for input $x$ has $O(\alpha(n))$ support,
\item if $s$ is a dominant SSE of $(V, \vec{P})$, then $s$ induces a dominant SSE in $(V', \vec{P})$,
\item $\lvert {u}_j(x, s, (V, \vec P)) - {u}_j(x, s, (V', \vec P)) \rvert < {1}/({4 \alpha(n))}$ for all $j \in \{1, \ldots, p(n)\}$, and
\item the utility gap guarantee is preserved, that is, if the answer bit under $s$ is wrong, then there exists a subform $H_{I}$ in the
game $(V', \vec{P})$ (reachable under $s$) and a prover $P_j$
acting at $H_I$, such that $P_j$ loses a $1/(2\alpha(n))$
amount in his expected payment compared to a strategy profile where $s_{I}$
(induced by $s$ on $H_I$) is replaced by $s_{I}^*$ (the dominant SSE on $H_I$),
keeping the strategy profile outside $H_I$, $s_{-I}$, fixed.
\end{itemize}
\end{lemma}

We prove~\lemref{pruning} in several parts. First, given an input $x$ and a strategy $s$ of the provers, we show how to transform
any verifier $V$ that imposes a probability distribution over outcomes with exponential support
into a verifier $V'$ that imposes a probability distribution with $O(\alpha(n))$ support.

Let $(V, \vec{P})$ use $p(n)$ provers and let the running time of $V$ be $n^{k}$ for some constant $k$.
There can be at most $2^{n^{k}}$ different payments that $V$ can generate for a
particular prover given input $x$. 
Given $x$ and $s$, fix a prover index $j \in \{1,\ldots, p(n)\}$. Let
${R}_1, {R}_2, \ldots, {R}_m$ be the payments generated by $V$ on $s$ for $P_j$. 
Let $V$'s randomness assign
probability distribution $\mu = (p_1, p_2, \ldots, p_m)$ to ${R}_1, {R}_2, \ldots,
{R}_m$ respectively.  Then, the expected payment of $P_j$ under $s$ is ${u}_j(x,s, (V, \vec
P)) = \sum_{i=1}^m p_i {R}_i$.

Recall that ${u}_j(x,s, (V, \vec P)) \in [-1,1]$ for all $1 \le j \le p(n)$. For each prover $P_{j}$, divide the interval
$[-1,1]$ into $4 \alpha(n)$ intervals, each of length $1/(2\alpha(n))$.  In
other words, prover $P_{j}$'s $i$th interval is
$[i/2\alpha(n),(i+1)/2\alpha(n))$,\footnote{To
include $1$ as a possible payment, interval $2\alpha(n)-1$ should be closed on
both sides; we ignore this for simplicity.}
 for each $i \in \{-2\alpha(n), \ldots,
2\alpha(n)-1\}$.

We round the possible payments for~$P_j$ to a representative of the their corresponding interval.
Specifically, we map each payment~$R_i$ to~$r_j$ as described in Equation~\ref{paymentmap}.
\begin{figure}[phtb]
\vspace*{-2pt}
\begin{minipage}{0.47\linewidth}
{\normalsize
\begin{equation}\label{paymentmap}
r_j = \left\{
\begin{array}{ll}
\frac{4\ell+1}{4\alpha(n)} & \mbox{if } R_i \in \left[ \frac{\ell}{2\alpha(n)}, \frac{2\ell+1}{4\alpha(n)}\right)\\ \\
\frac{4\ell+3}{4\alpha(n)} & \mbox{if } R_i \in \left[\frac{2\ell+1}{4\alpha(n)}, \frac{\ell+1}{2\alpha(n)} \right)\\
\end{array}
\right.
\end{equation}
}
\end{minipage}
\hspace*{0.2cm}
\begin{minipage}{0.47\linewidth}
{\normalsize

\begin{equation}\label{probmap}
p_i' = \left\{ \begin{array}{ll} \sum_{k \in T_j} p_k & \mbox{if } i = f(S(i))\\
		  0 & \mbox{otherwise }
	\end{array}
\right.
\end{equation}
}
\end{minipage}\figlabel{simple-nexp-ncmrip}
\vspace{-1em}
\end{figure}
There are potentially exponentially many different payments ${R}_i$, and only
polynomially many different payments ${r}_j$, so several ${R}_i$ must map to the same
${r}_j$.  Let $T_j = \{i : R_i \mbox{~maps to~} r_j\}$. Let $\mathcal{T} = \cup_j \{T_j\}$.
Thus the total number of distinct $r_j$ is $8 \alpha(n)$, so $|\mathcal{T}| = O(\alpha(n))$.
Let $S: \{1,\ldots,m\}\rightarrow {\mathcal T}$ such that $S(i) = T_j$ if and only if $i \in T_j$.

For each $T_j \in \mathcal T$, let $f(T_j)$ denote a unique index in
the set $T_j$. Without loss of generality, let $f(T_j)$ be the lowest index in $T_j$.
%
We define a new probability distribution $\mu'= (p_1', \ldots, p_h')$ over the payments
$R_1, \ldots, R_h$ respectively, given by~Equation~\ref{probmap}.
In particular, for every $T_j \in \mathcal{T}$, assign $R_{f(T_j)}$ probability $\sum_{k \in T_j}p_k$ and
for every other index $\ell \in T_j$, $\ell \neq f(T_j)$, assign $R_\ell$ probability $0$.

Define $V'$ as a polynomial-time verifier that simulates all deterministic computation
of~$V$. For a fixed input~$x$, $V'$ imposes a probability distribution $\mu'$
with $O(\alpha(n))$ support for any probability distribution $\mu$ imposed by $V$.
For other inputs, $V'$ simulates $V$ without any modification.

Note that given input $x$, a strategy
profile $s$ and the protocol $(V, \vec{P})$, transforming the distribution $\mu$ to $\mu'$
takes time linear in the size of the game tree, and thus exponential in $n$. (This
means that an $\sf{NEXP}$ oracle, given $x$, can guess a particular $s$ and perform the transformation.)

The remainder of the proof of \lemref{pruning} consists of the following three claims.  

First, we show that if the strategy profile $s$ is a dominant SSE of $(V, \vec{P})$,
then $s$ restricted to the pruned game tree of $(V', \vec{P})$ imposes
a dominant SSE on $(V', \vec{P})$ as well. 



%
\begin{cclaim}\label{validprotocol}
Any dominant SSE $s$ of the game formed by $(V, \vec{P})$ induces
a dominant SSE in the game formed by $(V', \vec{P})$.
\end{cclaim}

\begin{proof}
By contradiction, suppose $s$ is not an SSE of $(V', \vec{P})$. Then there exists an information set
$I= \{h_1, \ldots, h_m\}$, such that, conditioned on reaching $I$, the prover acting at $I$ can improve his
expected payment by deviating (given his belief $u_I'$ at $I$ if $I$ is reachable under $s$
and for any belief he may hold at $I$ if $I$ is unreachable under $s$).

We split into two cases: $I$ is either reachable or unreachable under $s$.

By construction, if $I$ is reachable under $s$ in $(V',\vec{P})$, then $I$ must also be reachable under
$s$ in $(V, \vec{P})$. Let $\mu_I' = (p_1', \ldots, p_m')$, where $p_i'$
is the probability assigned to $h_i$ and the support of $\mu_I'$ is
$O(\alpha(n))$. Let $R_1, \ldots, R_m$ be the payments that the player acting on $I$ gets under $s$
conditioned on reaching $h_1, \ldots, h_m$ respectively.
Similarly, let $R_1', \ldots, R_m'$ be the payments conditioned
on reaching $h_1, \ldots, h_m$ respectively under the strategy to which the player at
$I$ deviates from $s$. Then, $ \sum_{i = 1}^m p_i' R_i' >
\sum_{i =1}^m p_i' R_i$.  Let $\mu_I = (p_1, \ldots,
p_{m})$ be the beliefs on $I$ under $s$ in $(V, \vec{P})$. We use the
relationship between the distributions $\mu_I'$ and $\mu_I$, to show that such
a deviation in $(V', \vec{P})$ would imply a deviation in $(V, \vec{P})$. In
particular, mapping $\mu_I'$ back to $\mu_I$, using~Equation~\ref{probmap} we get:
\begin{align}
\sum_{i = 1}^m  \bigg( \mathbb{I}_{i =f(S(i))}  \cdot \sum_{k \in S(i)} p_k \bigg)R_i'
&> \sum_{i =1}^m \bigg(  \mathbb{I}_{i =f(S(i))}  \cdot \sum_{k \in S(i)} p_k \bigg) R_i \notag \\
\sum_{i = 1}^m \bigg( \mathbb{I}_{i =f(S(i))}  \cdot \sum_{k \in S(i)} p_k \bigg) \cdot \min_{k \in S(i)} R_k'
&> \sum_{i = 1}^m \bigg( \mathbb{I}_{i =f(S(i))} \cdot \sum_{k \in S(i)} p_k \bigg) \cdot \max_{k \in S(i)} R_k\label{ineq:trans} \\
\sum_{i = 1}^m \bigg( \mathbb{I}_{i =f(S(i))}  \cdot \sum_{k \in S(i)} p_k  R_k'\bigg)
&> \sum_{i = 1}^m \bigg( \mathbb{I}_{i =f(S(i))}  \cdot \sum_{k \in S(i)} p_k R_k\bigg)\notag\\
\sum_{i=1}^{m} p_i R_i' &>\sum_{i=1}^{m} p_i R_i\label{ineq:final}
\end{align}
Inequality~\ref{ineq:trans} holds because $R_{f(S(i))}' > R_{f(S(i))}$,
and so the two payments lie in different intervals in the mapping (Equation~\ref{paymentmap}). Thus the minimum payment in the interval
of $R_{f(S(i))}'$ will be greater than the maximum payment in the interval of $R_{f(S(i))}$.
Finally, Inequality~\ref{ineq:final} contradicts the fact that $s$ was an SSE in $(V, \vec{P})$, achieving a contradiction for the case of reachable information sets.

For unreachable information sets the argument is easy. If $I$ is unreachable under $s$ in $(V', \vec{P})$, then $I$ must be unreachable under $s$ in
$(V, \vec{P})$.
If the action of prover acting at $I$ is not his best response in $(V', \vec{P})$ for some history $h \in I$ then,
it contradicts the fact that $s$ is an SSE of $(V, \vec{P})$.

Now, suppose $s$ is not a dominant SSE of $(V', \vec{P})$.
Then there exists a subgame $H_I$ of height $k$ such that $s$ is
dominant on all subgames following $H_I$ of height $<k$ but not
weakly-dominant at $H_I$ (among SSE's that are dominant at all subforms
following $H_I$). Let $s^*$ be dominant on $H_I$, then
the expected payment of at least one prover $P_j$ is better under $s^*$, while everyone else
does just as well (given the beliefs at $I$ derived using Bayes' rule if $I$ is reachable
under $s$ or given any
beliefs if $I$ is unreachable under $s$). Writing out the expression of expected payment
of $P_j$ conditioned on reaching $H_I$ and ``unfolding'' the probability distribution back
to the original game, we get a contradiction that $s$ could not have been
a dominant SSE of the original game, as the same strategy $s^*$ would give $P_j$ a better
expected payment at $H_I$ while doing as well for other provers.
The proof is similar to the above and we omit the details.
\end{proof}

The following claim states that for a given $s$, the expected payments of the provers
under $(V, \vec{P})$ and under $(V', \vec{P})$ are not too far off.  This claim is one of the bullet points in \lemref{pruning}, and will be used to prove Claim~\ref{gapsame}.

\begin{cclaim} \label{paydiff}
For all $j \in \{1, \ldots, p(n)\}$, $\lvert {u}_j(x, s, (V, \vec P)) - {u}_j(x,  s, (V', \vec P)) \rvert < {1}/{(4 \alpha(n)})$.
\end{cclaim}

\begin{proof}  
%
%
Given input $x$ and strategy profile $s$, fix a~prover~$P_j$.
Let $V$ generate payments ${R}_1, {R}_2, \ldots, {R}_m$ under $s$ for $P_j$,
and assign the probability distribution $\mu = (p_1, p_2, \ldots, p_m)$ on ${R}_1, {R}_2, \ldots,
{R}_m$ respectively.
Using~Equations~(\ref{paymentmap}) and~(\ref{probmap}) we compare $P_j$'s expected payment:
\begin{align*}
&\lvert {u}_j(s,x, (V, \vec P)) - {u}_j(s,x, (V', \vec P)) \rvert
= \bigg| \sum_{i=1}^m p_i {R}_i - \sum_{T_j \in \mathcal T}   \bigg( \sum_{k \in T_j} p_k\bigg) {R}_{f(T_j)}  \bigg| \\
=&\sum_{T_j \in \mathcal T} \sum_{k \in T_j} p_k \bigg( |  {R}_{f(T_j)}   -  {R}_i |\bigg) 
< \sum_{T_j \in \mathcal T}    \sum_{k \in T_j} p_i \bigg( \frac{1}{4 \alpha(n)}\bigg) = \bigg(\sum_{i=1}^m p_i \bigg)  \frac{1}{4 \alpha(n)} =\frac{1}{4 \alpha(n)}\qedhere 
\end{align*}
%

\end{proof}

To complete the proof of~\lemref{pruning}, 
we show that $(V', \vec{P})$ preserves
utility gap guarantees.

\begin{cclaim}\label{gapsame}
Given input $x$, if the answer bit under $s$ is wrong,
then there exists a subform $H_I$ reachable under $s$ in $(V', \vec{P})$ and $P_j$
acting at $H_I$, such that $P_j$'s expected payment under~$s$ is $\frac{1}{2\alpha(n)}$
less than his expected payment under~$(s_{-I},s_I^*)$, where
$s_{I}^*$ is a dominant SSE on $H_I$.
\end{cclaim}

\begin{proof} 
Consider a strategy profile $s^*$ that is a dominant SSE in the game
tree of $(V, \vec{P})$. Since $s$ gives the wrong answer bit, from the $\alpha(n)$-utility gap
guarantee of $(V, \vec{P})$ and~\defref{rewardgap}, there exists a subform $H_I$ reachable under $s$,
such that a prover $P_j$ acting in $H_I$ loses $1/\alpha(n)$
in his expected payment under $s$ compared to the strategy profile~$(s_{-I},s_I^*)$. That is,

\begin{equation}\label{claimineq:1}
 {u}_j(x, (s_{-I}, s_{I}^*), (V, \vec P)) - {u}_j(x, (s_{-I}, s_{I}), (V, \vec P))  > \frac{1}{\alpha(n)}.
\end{equation}

Using Claim~\ref{validprotocol}, $s^*$ also induces a dominant SSE in the game tree of $(V', \vec{P})$.
And since $H_I$ is reachable under $s$ in $(V, \vec{P})$, it is reachable under $s$ in $(V', \vec{P})$ as well.
We show that:
\begin{equation}\label{claimineq:4}
 {u}_j(x, (s_{-I}, s_{I}^*), (V', \vec P)) - {u}_j(x, (s_{-I}, s_{I}), (V', \vec P))  > \frac{1}{2\alpha(n)}.
\end{equation}

Using~Claim~\ref{paydiff}, prover $P_j$'s expected payments in the two protocols under~$s$ and $s^*$ follow:
\begin{align}
\lvert {u}_j(x, (s_{-I}, s_{I}^*), (V, \vec P)) - {u}_j(x, (s_{-I}, s_{I}^*), (V', \vec P)) \rvert &< \frac{1}{4 \alpha(n)}\label{claimineq:2}\\
\lvert {u}_j(x, (s_{-I}, s_{I}), (V, \vec P)) - {u}_j(x, (s_{-I}, s_{I}), (V', \vec P)) \rvert &< \frac{1}{4 \alpha(n)}\label{claimineq:3}
\end{align}

There are four cases depending on the sign of the left hand side of
Inequalities~(\ref{claimineq:2}) and~(\ref{claimineq:3}).
We show that~Claim~\ref{gapsame} holds for one of the cases and omit the details of the others, which are similar.

Suppose the left hand side of both inequalities is positive, that is,
$u_j(x, (s_{-I}, s_{I}^*), (V, \vec P)) >
u_j( x, (s_{-I}, s_{I}^*), (V', \vec P))$, and
$u_j(x, (s_{-I}, s_{I}), (V, \vec P)) >u_j(x, (s_{-I}, s_{I}), (V', \vec P))$. Then
\begin{align*}
&u_j(x, (s_{-I}, s_{I}^*), (V', \vec P)) - u_j(x, (s_{-I}, s_{I}), (V', \vec P))\\
&\qquad\qquad> \bigg(u_j(x, (s_{-I}, s_{I}^*), (V, \vec P)) - \frac{1}{4\alpha(n)}\bigg) - u_j(s', x, (V', \vec P))\\
&\qquad\qquad> \bigg(u_j(x, (s_{-I}, s_{I}), (V, \vec P))  + \frac{1}{\alpha(n)} \bigg) -\frac{1}{4\alpha(n)}- u_j(x, (s_{-I}, s_{I}), (V', \vec P))
> \frac{3}{4\alpha(n)}. \qedhere
\end{align*}
\end{proof}

Using~\lemref{pruning}, given an $O(\alpha(n))$-gap ncRIP protocol (where $\alpha(n)$ is constant or polynomial),
a polynomial-time oracle Turing machine can
use its $\sf{NEXP}$ oracle to guess a strategy
profile $s$, prune the verifier's Nature moves, and report the
resulting $O(\alpha(n))$-support distribution bit-by-bit.
Thus, it can simulate the new distribution
and find the decision nodes that are reachable under $s$.

\paragraph*{Searching through the strategy-profile space efficiently}
The next question is: how should the polynomial-time Turing machine
navigate the potential strategy-profile space (in polynomial time) to find the
strategy profile that satisfies~Observation~\ref{gaprestate}?
To cut down on the
search space, we invoke a recurring idea: divide each prover's expected payment interval $[-1,1]$,
evenly into $8 \alpha(n)$ \defn{subintervals} of length $1/(4\alpha(n))$, and
consider \defn{subinterval profiles} (a tuple of subintervals, one for each prover).
\begin{lemma}\label{intervalsse}
Given an input $x$ and an ncRIP protocol $(V, \vec{P})$ with $\alpha(n)$-utility gap,
consider a {subinterval profile} $(L_1, \ldots, L_{p(n)})$, where
each $L_i = [{k}/({4\alpha}), ({k+1})/({4 \alpha +1}))$ denotes a subinterval of prover $P_i$ in $[-1,1]$,
for some $k \in \{-2\alpha(n), \ldots,
2\alpha(n)-1\}$.
Let $s$ be an SSE that has an expected payment profile $\tilde{u}(x,s)$ such that $u_i(x,s) \in L_i$ for all $1 \le i \le p(n)$, and
$s$ does not satisfy~Observation~\ref{gaprestate}. Then the expected payment profile $\tilde{u}(x,s^*)$ under a dominant SSE
$s^*$ cannot lie in the same subinterval profile, that is, there exists
a prover index $j$ such that $u_j(x, s^*) \notin L_j$.
\end{lemma}

\begin{proof} 
Since $s$ does not satisfy~Observation~\ref{gaprestate},
there exists a reachable subform~$H_I$ and prover~$P_j$
acting on $H_I$ such that the following holds. Without loss of generality, let $\mu_j(s,x) \in L_k$.
\begin{align*}
&u_j (x, (s_{-I}, s_{I}^*), (V, \vec{P})) -  u_j (x, (s_{-I}, s_{I}), (V,\vec{P})) > \frac{1}{\alpha(n)}\\
&u_j (x, s^* , (V, \vec{P})) >  \frac{1}{\alpha(n)} +\frac{k}{4 \alpha(n)} \implies u_j (x, s^* , (V, \vec{P}))  \notin L_k\qedhere
\end{align*}
\end{proof}

Using Lemma~\ref{intervalsse}, if the polynomial-time Turing machine is able to test
{\em any} SSE $s$ with $\tilde{u}(x,s)$ in a subinterval profile, 
for all subinterval profiles, then it is guaranteed to find
one that satisfies~Observation~\ref{gaprestate}. This is because a dominant SSE of
an ncRIP protocol is guaranteed to exist and its expected payment profile must belong to some subinterval profile.

However, there are still $O(\alpha(n))$ subintervals for each prover, and thus $O(\alpha(n)^{p(n)})$ total subinterval profiles. A polynomial-time machine cannot test SSEs for each of them.



%
To reduce the search space further, we show that it is sufficient to consider
subintervals of the {total expected payment} rather than individual and test an SSE $s$
for each of them.  
%
Recall that a SSE $s$ is weakly dominant if for any player~$i$ and SSE~$s'$,
$u_i(s)\geq u_i(s')$.  

\begin{lemma}\lemlabel{maxsse}
If a weakly-dominant SSE exists, then a strategy profile $s$ is a weakly-dominant SSE if
and only if $s$ is an SSE and $s$ maximizes the sum of utilities of all players among all
SSEs.
\end{lemma}


We are now ready to prove the upper bound for ncRIP classes with constant, polynomial,
and exponential utility gap.

\paragraph*{Constant utility gap}
%
Using~\lemref{pruning} and~\lemref{maxsse}, simulating a constant-gap protocol
using a $\sf{P^{\sf{NEXP}[O(1)]}}$ machine $M$ is straightforward. We give
a high-level overview below.

There are at most $O(1)$ subforms that are reachable under any strategy profile $s$,
and the total expected payment of the provers conditioned on
reaching these subforms will be in one of the $O(1)$ subintervals.
Thus, there are $O(1)$~combinations of total expected payments on all subforms
(including the whole game). $M$ queries its $\sf{NEXP}$ oracle whether there exists an SSE that
achieves that combination of total expected payments on those subforms, for all combinations.
%

\begin{lemma}\lemlabel{const-upper}
$\cncRIP \subseteq \sf{P^{NEXP[O(1)]}}$.
\end{lemma}

\begin{proof} 
Given any $L \in \ancRIP$, let $(V, \vec{P})$ be the
MRIP protocol with $\alpha(n)$ utility gap for~$L$, where $\alpha(n)$ is a constant. 

Given an input $x$ of lenth $n$, consider the following deterministic
polynomial-time oracle Turing machine $M$ with access to an oracle $O$ for an
$\sf{NEXP}$ language. Similar to the proof
of~\lemref{poly-upper}, $M$ divides $[-1,1]$ into $8 \alpha(n)$ intervals, each of
length $1/4\alpha(n)$.  In other words, the $i$th interval is
$[i/4\alpha(n),(i+1)/4\alpha(n))$ for each $i \in \{-4\alpha(n), \ldots,
4\alpha(n)-1\}$.\footnote{To include $1$ as a possible reward, interval
$4\alpha(n)-1$ should be closed on both sides; we ignore this for simplicity.}

Using~\lemref{pruning}, under a given input $x$ and strategy profile $s$,
there are at most $8 \alpha(n)$ subforms are reached under any $s$ in the modified game.
Total expected payment of provers
acting within any subform (conditioned on reaching the subform) must lie
in any one of the $8\alpha(n)$ intervals in $[-1,1]$. Thus overall, there are
$O(\alpha(n)^\alpha(n))$ combinations of total expected payments over subforms,
which is still $O(1)$. Let $(u, u_{I_1}, \ldots, u_{I_k})$ be a tuple of total expected payments, 
where $k = 8 \alpha(n)$, the maximum number of subforms reachable under any $s$,
and $u$ represents the total expected payment of the whole game, whereas $u_{I_j}$
represents total expected payment of the provers acting in subform $I_j$ (conditioned
on reaching $I_j)$.

For each combination $(u, u_{I_1}, \ldots, u_{I_k})$, $M$ queries $O$: {\em does
there exists a strategy profile that is an SSE and the total expected payments
over reachable subforms under $s$ and $O(\alpha(n))$ support
Nature moves imposed by~\lemref{pruning}
is $(u, u_{I_1}, \ldots, u_{I_k})$ (conditioned on reaching the subforms)?} 
Among the queries to which the oracle's answer is ``yes'',
$M$ finds the combination that achieves maximum total expected payment for all subforms.
Such a combination is guaranteed to exist because $(V, \vec{P})$ is an ncRIP protocol,
and a dominant SSE of the game exists. 
\end{proof}

\begin{remark} \label{rem:nonadaptiveconst}
The polynomial-time oracle Turing machine in~\lemref{const-upper} can
issue all its queries {\em non-adaptively}. 
    That is, $\ancRIP \subseteq \sf{P_{||}^{NEXP[O(1)]}}$. Furthermore in~\secref{lower}, we show that $\cncRIP \subseteq \sf{P^{NEXP[O(1)]}}$.
    Indeed, the two classes are equal: $\sf{P^{NEXP[O(1)]}_{||}}= \sf{P^{NEXP[O(1)]}}$.

    Since $\cMRIP = \sf{P_{||}^{NEXP[O(1)]}}$~\cite{ChenMcSi16,CMS17arxiv}, this shows that
cooperative provers are as powerful as non-cooperative provers under constant utility-gap guarantees,
and we obtain \corref{c-mrip-ncrip}.
\end{remark}

\paragraph*{Polynomial utility gap}
Next, we prove the upper bound of the case of polynomial utility gap. We note that the simple strategy of querying all possible
payment combinations as in~\lemref{const-upper} does not work (there
are $O(\alpha(n)^{\alpha(n)})$ total combinations).

To simulate a polynomial-utility gap ncRIP protocol $(V, \vec{P})$, using
a $\sf{P^{NEXP}}$~machine $M$, we put to use all the structure we have established
in this section. 


For each of the $O(\alpha(n))$ total payment subintervals of the interval $[-1,1]$
that correspond to an SSE, $M$ does a recursive search to find an exact total expected payment $u(x,s)$ that is generated by an SSE. 
(We can restrict ourselves to $O(\alpha(n))$ oracle queries due to \lemref{maxsse}.)
In particular, $M$ queries the $\sf{NEXP}$ oracle: {\em Does there exist an
SSE with total expected payment in the first half of the $i$th interval?}.
If the answer is {\em yes} then $M$ recurses on the first half of the $i$th interval; $M$ does not need to search the second half by Lemma~\ref{intervalsse}.  Otherwise (if the answer is {\em no}) then $M$ recurses on the second half.
Thus,
in polynomial time and with polynomial queries, $M$ can find an exact $u(x,s)$ for an SSE $s$ in the subinterval using the power of its  {\em adaptive} queries.

Next, $M$ simulates the protocol $(V, \vec{P})$
with the help of the oracle, under the SSE $s$ for a given subinterval.
\lemref{pruning} is crucial for $M$ to
simulate the verifier's moves, because $V$ in general can induce
exponential-size distributions.
$M$ traverses the tree reachable under $s$ ``top-down'' using the oracle to learn the pruned distributions and provers' moves.
Finally, $M$ goes ``bottom-up'' to test whether $s$ satisfies~Observation~\ref{gaprestate} on all its reachable subgames.

%
%

\begin{lemma}\lemlabel{poly-upper}
$\pncRIP \subseteq \sf{P^{NEXP}}$.
\end{lemma}

\begin{proof} 
Given any $L \in \pncRIP$, let $(V, \vec{P})$ be the
ncRIP protocol with $\alpha(n)$ utility gap for~$L$, where $\alpha(n)= n^k$ for
some constant $k$. 

Given an input $x$ of lenth $n$, consider the following deterministic
polynomial-time oracle Turing machine $M$ with access to an oracle $O$ for an
$\sf{NEXP}$ language. $M$ divides $[-1,1]$ into $8 \alpha(n)$ intervals, each of
length $1/4\alpha(n)$.  In other words, the $i$th interval is
$[i/4\alpha(n),(i+1)/4\alpha(n))$ for each $i \in \{-4\alpha(n), \ldots,
4\alpha(n)-1\}$.\footnote{To include $1$ as a possible reward, interval 
$4\alpha(n)-1$ should be closed on both sides; we ignore this for simplicity.} 

For each interval $[i/4\alpha(n),(i+1)/4\alpha(n))$, $M$ makes the following queries
to $O$: {\em does there exist a strategy profile ${s}$ that is an SSE and the
sum of expected payments of all provers $u(x, s)$ is in the $i$th interval?}
Let $L$ denote the set of intervals for which the answer to the query is
``yes''. 

For each interval $[\ell/4\alpha(n), (\ell+1)/4\alpha(n)) \in L$,  $M$ queries $O$: {\em does there exist a strategy profile $s$ that is
an SSE and the sum of expected payments of all provers $u(x, s)$ is in the
first half of the $\ell$th interval?} If the answer is ``yes'', then $M$
recurses on the first half, else $M$ recurses on the second half of the
interval.  In polynomial time and polynomial queries, $M$ can find the exact
total expected payment $u(x, s, (V, \vec{P}))$ in the interval that is
generated by an SSE. $M$ asks further queries to figure out the exact payment
profile under such an SSE. For $k \in \{1,\ldots, p(n)\}$, where $p(n)$ is the
total number of provers in $(V, \vec{P})$, and for each $j \in \{1, \ldots,
n^{k'}\}$, where $n^{k'}$ is the running time of $V$ ($k'$ is a constant), $M$
asks the following queries adaptivily: {\em under an SSE where
$\sum_{i=1}^{p(n)} \mu_i (x,s) = u(x, s)$, what is the $j$th bit in the
expected payment $\mu_k(x, s)$ of prover $P_k$, given  and the first $j-1$ bits
of $\mu_k(x,s)$ and $\mu_1(x,s), \ldots, \mu_{k-1}(x,s)$}.  In $O(n^{k'} p(n))$
queries, $M$ can figure out the exact payment profile $\tilde{u}(x,s)=(\mu_1,
(x, s) \ldots, \mu_k (x,s))$ under an SSE $s$, such that the total expected
payment is in the $\ell$th interval.

$M$ now verifies whether the SSE corresponding to the payment profile
$\tilde{u}(x,s)$ satisfies the condition of Observation~\ref{gaprestate}. $M$
proceeds in two phases: first, $M$ wants to go ``top-down'' figuring out what
part of the game tree is being played under $s$ on input $x$, using the oracle
to simulate the provers and the verifier.  Then, it goes ``bottom-up'' in the
tree being played under $s$, to check whether all subforms are
``$(1/\alpha(n))$-close'' to the dominant strategy at that subform.

\paragraph*{Top-down phase.} Let $k(n)$ be the total number of rounds in $(V,
\vec{P})$. Note that $k(n)$ is polynomial in $n$.  Let $m_{ij}$ denote the
message sent by prover $P_i$ at round $j$.  Then, for each round $j$ and each
prover $i$ where $1\le j \le k(n)$ and $1 \le k \le p(n)$, $M$ first asks the
oracle to give the ``pruned'' $O(\alpha(n))$ support distribution imposed by the
Nature move of $V$ at round $j$ bit by bit as follows: {\em ``under an SSE
where the expected payment profile is $\tilde u (x,s)$, what is the $r$th bit
of the distribution imposed by $V'$~using $V$ and~\lemref{pruning}?''} 
This requires a polynomial number of bits (and therefore queries) because the distribution is polynomial sized.
The pruned distribution preserves the dominant SSE and changes the utility gap
by only a factor $2$ (this factor does not affect the proof
as our intervals are scaled down to handle it).
Given this distribution, $M$ simulates $V$ on the support of the distribution to figure out the messages that $V$ sends
to the provers in round $j$. 
In particular, $M$ does not have access to random bits, so instead it simulates \emph{every} action of $V$ in the support.
To simulate the provers at round $j$, $M$
similarly queries $O$ bit by bit: {\em ``under an SSE where the expected
payment profile is $\tilde u (x,s)$, what is the $r$th bit of the message sent
by $P_k$''}.  Thus, after simulating the moves of $V$ and $P$ under $s$, $M$
has sketched out the $O(\alpha(n))$ size part of the game tree being played under
$s$ corresponding to $\tilde u (x,s)$.

\paragraph*{Bottom-up phase.} 
Given the $O(\alpha(n))$ nodes of the game tree under play,
$M$ can mark out the subforms reachable under $s$ corresponding to
$\tilde{u}(x,s)$. Going from the last level up, for each subform
$H_I$ reachable under $s$, $M$ uses the oracle to figure
out which payment interval the expected payments of the weakly-dominant SSE 
on $H_I$ lie in (given the expected weakly-dominant SSE payments on the reachable
subforms verified so far), until it finds a subform that
violates the condition of Observation~\ref{gaprestate}. 

In particular, for each subform $H_I$ of height $k$, let $\tilde{u}(x, s, I')$
denote the tuple of total expected payments under $s$ on all subforms $H_{I'}$ of height $<k$ 
following $I$ (conditioned on reaching $I$) verified so far.
$M$ divides the interval $[-1,1]$ into
$8 \alpha(n)$ intervals of size $\alpha(n)/4$ as before and for each interval
queries the oracle $O$: {\em does there exist a strategy profile ${s_I}$ on
subagme $H_I$ that is an SSE and the
sum of expected payments of all provers $u(x, s, I)$ is in the $x$th interval,
and gets a total expected payments on subforms $H_{I'}$ of height $<k$ following $I$ equal to $\tilde{u}(x, s, I')$}.\footnote{$M$ does not {\em need}
to send the 
total expected payments of the subforms at lower levels. Instead,
$M$ can just send the total expected payment $u(x, s)$ at the root and ask $O$ to guess $s$ as well.
An $\sf{NEXP}$ can verify if one SSE weakly dominates another. This observation is crucial in extending this proof to exponential utility gap.}

Then, $M$ finds the maximum interval $[i/4\alpha(n), (i+1)/4\alpha(n))$ among the intervals for which the oracle says yes.
By \lemref{maxsse}, the weakly-dominant SSE $s_I^{\mbox{max}}$ at $H_I$
also lies in the $i$th interval. Using the probability $p_I$ assigned by $H_I$ ($M$ knows the
distribution imposed by all ``pruned'' Nature moves), $M$ checks whether the total
expected payment of weakly-dominant SSE $s_I^{\mbox{max}}$ is in the same interval as the
sum of expected payments of provers in $Z_I$ under $s$. If it is not, then $s$ fails the test and $M$ continues
to the next interval in $L$. Otherwise, $M$ continues to the next reachable subform. 

If $s$ passes the test  
for all subforms (including at the root), then by Observation~\ref{gaprestate},
the answer bit under $s$ is correct. 
$M$'s final query to $O$ is: {\em ``under an SSE
where the expected payment profile is $\tilde u (x,s)$, what is the answer bit $c$?} If $c=1$,
then $M$ accepts $x$, otherwise $M$ rejects $x$.

$M$ is guaranteed to find a payment profile $\tilde u(x,s)$ 
(and thus a strategy profile $s$) that passes the test. Since
$(V, \vec{P})$ is an ncRIP protocol for $L$, there exists a dominant 
SSE $s^*$ in some interval in $L$. By Obversation~\ref{gaprestate}, 
if a strategy profile $s'$ fails the test, the dominant SSE can not
get a total expected payment in the same interval as $s'$. Thus, we can rule out
intervals by checking any SSE with total expected payment in that interval.
Since a dominant SSE $s^*$ exists, $M$ must eventually find an interval,
where the corresponding SSE passes the test.

To complete the proof, we note that (a) $M$ runs in polynomial time, (b) each
query to the oracle is polynomial, and, (c) the oracle queries can be
answered in non-deterministic exponential time.

First, (a) holds because each top-down and bottom-up phase is executed $O(\alpha(n))$
times and each of the phases take polynomial time. In the top-down phase,
$M$ simulates the protocol on strategy $s$ using
the oracle while restricting the verifier's Nature moves
to be of $O(\alpha(n))$ support. Thus this phase takes polynomial time.
For the bottom-up phase, $M$ finds weakly-dominant SSEs at each reachable subforms under $s$.
Since there are at most $O(\alpha(n))$ subforms and at most $O(\alpha(n))$ interval
queries for each subform, the bottom-up phase takes time polynomial in $n$.

Second, (b) holds each oracle query involves a
total expected payment $\tilde{u} (x,s)$ or an interval of size $\alpha(n)/2$, 
both of which can be generated by $V$ and hence are polynomial in $n$.

To prove (c), it is sufficient to show that an $\sf{NEXP}$ machine can 
guess a strategy profile and verify if it is an SSE and if it
gets expected payments in a certain interval. Since
the transcript of any ncRIP protocol is polynomial in $n$, a strategy
profile $s$ of the provers can be represented in exponential bits,
and thus $O$ can guess such an $s$.
Now given $s$ and the protocol $(V,
\vec{P})$, by~\lemref{verify-sse}, it is possible to verify
whether ${s}$ is an SSE of the game in time linear in the size of the
game tree, and thus exponential in $n$.
Furthermore, it can compute the expected payments of the provers under
${s}$ in exponential time as well, which is sufficient to answer all the queries made by $M$.
\end{proof}

\paragraph*{Exponential utility gap}
We conclude by giving a tight upper bound on the class of ncRIP
protocols with exponential utility gaps. The proof
follows immediately from that of~\lemref{poly-upper}.
In fact, it is simpler as the exponential-time Turing machine
is powerful enough to (a) simulate $V$'s Nature moves directly, and (b)
test all possible payment profiles. Thus, in the case of exponential
utility gap, we do not need~\lemref{pruning} or the notion of subintervals.


\begin{lemma}\lemlabel{exp-upper}
$\sf{ncRIP} \subseteq \sf{EXP^{poly-NEXP}}$.
\end{lemma}

\begin{remark}\label{rem:nonadaptiveexp}
    Since $\sf{EXP^{poly-NEXP}} \subseteq \sf{EXP_{||}^{poly-NEXP}} = \sf{EXP_{||}^{NP}}$,
    and $\sf{EXP_{||}^{NP}} \subseteq \sf{MRIP}$~\cite{ChenMcSi16},
\lemref{exp-upper} shows that $\encRIP \subseteq \eMRIP$ and using~\lemref{mrip-ncrip},
we get that in general the two classes coincide. In other words, non-cooperative rational proofs are
 as powerful as cooperative multi-prover rational proofs under exponential utility gap
and we obtain \corref{exp-mrip-ncrip}.
\end{remark}
 
\section{Additional Related Work}\seclabel{additionalrelated}

\paragraph*{Rational Proofs}
The model of single-prover \defn{rational interactive proofs} (RIP) was introduced by Azar and Micali~\cite{azar2012rational},
who used scoring rules as the main tool to construct simple and efficient RIP protocols.
In a follow-up work~\cite{azar2013super}, they extended this work to design super-efficient rational proofs
that have sublinear verification and computation compelexity.
Guo et al. present rational \defn{arguments}
for a computationally bounded prover and a sublinear verifier in~\cite{guo2014rational},
and construct
rational arguments for all languages in $\sf{P}$~\cite{guo2016rational}.
Campanelli and Rosario~\cite{campanelli2015sequentially} study sequentially composable rational proofs
and rational proofs for space bounded computations~\cite{campanelli2017efficient}, 
while Zhang and Blanton~\cite{zhang2014efficient} design protocols to outsource matrix multiplications to a rational cloud.

The model of \defn{multi-prover (cooperative) rational interactive proofs} (MRIP) 
was introduced by Chen et al.~\cite{ChenMcSi16}. In this model, the provers work together to maximize their \emph{total payment}.
They show that the class equals $\sf{EXP^{||NP}}$ under exponential utility gap and $\sf{P^{||NEXP}}$ under polynomial utility gap. In the full version~\cite{CMS17arxiv},
they show that MRIP under constant utility gap is equal to $\sf{P^{||NEXP[O(1)]}}$. In follow-up work~\cite{ChenMcSi18},
the authors scale down the power of the verifier and design super-efficient MRIP protocols with strong utility-gap guarantees.

\paragraph*{Game-Theoretic Characterization of Complexity Classes} Game-theoretic
characterization of complexity classes has been largely studied in the form
of \defn{refereed games}~\cite{chandra1976alternation, feige1990noisy,
feige1997making, feige1992multi, reif1984complexity,
feigenbaum1995game,koller1992complexity}.
Chandra and Stockmeyer~\cite{chandra1976alternation} show that any language in
$\sf{PSPACE}$ is refereeable by a game of perfect information.
Feige and Kilian~\cite{feige1997making} show that the class of imperfect information,
perfect recall refereed games is exactly $\sf{EXP}$.
Feigenbaum, Koller and Shor~\cite{feigenbaum1995game} 
show that if provers are allowed to have imperfect recall (essentially acting as oracles),
refereed games can simulate $\sf{EXP^{NP}}$.

\paragraph*{Query Complexity and Related Complexity Classes}
The query complexity of oracle Turing machines
has been widely studied in the literature~\cite{beigel1991bounded,wagner1990bounded,buhrman1999quantum}.
In this paper, we give game-theoretic characterizations of the classes $\sf{P^{NEXP[O(1)]}}$. $\sf{P^{NEXP}}$, and $\sf{EXP^{poly-NEXP}}$.

\section{Properties of Strong Sequential Equilibrium}
\seclabel{strongse}

In this section, we prove several important
properties of strong sequential equilibrium, which make
it a good candidate solution concept in designing extensive-form mechanisms.



\paragraph*{Strong sequential equilibrium admits a sequential equilibrium}
We first show that,
given a strategy profile ${s}$ that is a strong sequential equilibrium (thus does not rely on a belief system),
we can construct
a belief system ${\mu}$ such that the pair $({s}, {\mu})$
forms a sequential equilibrium.

\begin{lemma}\lemlabel{lem:SE}
For any strategy profile ${s}$ that is a strong sequential equilibrium, there exists
a belief system ${\mu}$ such that $({s}, {\mu})$ is a sequential equilibrium.
\end{lemma}

\begin{proof}
The sequential-rationality requirement
will follow easily from the definition of SSE.
To
prove that $s$ admits a sequential equilibrium,
the key is to  pair it with a consistent belief
system; see~\secref{ncmrip} for definition.
Indeed, we construct a belief system $\mu$ and show that,
 there exists a
sequence of pairs $(s^\epsilon, \mu^\epsilon)_{\epsilon\rightarrow 0}$ which converges to $(s,\mu)$,
as $\epsilon$ goes to $0$,
where each $s^\epsilon$ is a profile of
completely mixed behavioral strategies
 and each
 $u^\epsilon$ is the belief system
derived from $s^\epsilon$ using Bayes' rule.

Recall that a strategy profile ${s}$ defines a probability distribution over the actions available to
a player at an information set where he acts.
That is, for each information set $I_i$ of a player $i$, $s_i(I_i)$
is a probability
distribution
over $A(I_i)$, the set of actions available to player $i$ at $I_i$.
In particular, if $A(I_i) = (a_1, \ldots, a_k)$,
then $s_i(I_i) = (p_i(a_1), \ldots, p_i(a_k))$ where
$p_i(a_\ell)$ is the probability that player~$i$
chooses action $a_\ell$ at $I_i$.
Let $A^+(I_i)$ and
$A^0(I_i)$
be the set of actions at information set $I_i$
which player $i$ chooses with
positive probability and zero probability respectively;
that is, $A^+(I_i) =
\{a_\ell \in A(I_i) \mbox{ $|$ } p_i(a_\ell) >0\}$ and
$A^0(I_i) = A(I_i) \setminus A^+(I_i)$.
For any $\epsilon\in (0, 1)$, we define $s^\epsilon_i$ for player $i$ at information set $I_i$ as follows:
if $A^0(I_i) = \emptyset$ then $s^\epsilon_i(I_i) = s_i(I_i)$; otherwise,
\[
s^\epsilon_i(I_i)(a_\ell) =
\left\{
	\begin{array}{ll}
 (1-\epsilon) \cdot p_i(a_\ell) & \mbox{for each $a_\ell \in A^+(I_i)$}; \\
 \frac{\epsilon}{|A^0(I_i)|} & \mbox{for each $a_\ell\in A^0(I_i)$}.
\end{array}
\right.
\]
By construction,
$s^\epsilon_i(I_i)$ is a valid probability distribution over
$I_i$ and
is completely mixed, that is, assigns a positive probability to every
action in $I_i$.
Indeed, because $\sum_{\ell = 1}^k p_i(a_\ell) = \sum_{a_\ell \in A^+(I_i)} p_i(a_\ell) = 1$,
when $A^0(I_i)\neq \emptyset$
we have
$\sum_{a_\ell \in A(I_i)} s^\epsilon_i(I_i)(a_\ell) = \sum_{a_\ell \in A^+(I_i)} (1-\epsilon) p_i(a_\ell) + \epsilon = 1$.
It is easy to see that $s^\epsilon_i$ converges to $s_i$ when $\epsilon \rightarrow 0$.

Given the strategy profile $s^{\epsilon}$,
to define $\mu_i^\epsilon$, the belief system of a player $i$,
consider an arbitrary information set $I_i$
 where player $i$ acts.
The probability
that a particular history $h= (a^1, \ldots, a^K)\in I_i$ occurs
can be derived from $s^\epsilon$ as follows.
For any history $h'=(a^1, \ldots, a^w)$ with $0\leq w\leq K-1$,
recall that $Z(h')$
is the player acting following
history $h'$.
For any
action $a \in A(h')$,
let $s_{Z(h')}^\epsilon (h')(a)$
denote the probability
assigned by $s_{Z(h')}^\epsilon$ to
action $a$ at history $h'$ (i.e., at the information set containing $h'$).
We have
\[
\prob{h \mbox{ occurs under } s^\epsilon} = \prod_{w=0}^{K-1} s_{Z(a^1, \ldots, a^w)}^\epsilon(a^1, \ldots, a^w)(a^{w+1})
= c_h \epsilon^{e_h} (1-\epsilon)^{f_h},
\]
where $c_h, e_h$ and $f_h$ are positive constants depending on $s$ and $h$, but not on $\epsilon$.
In particular, letting $S^0$
be the set of actions
$a^{w+1}$
in $h$ that are assigned
zero probability by $s_{Z(h')}$
at history $h' = (a^1,\dots, a^{w})$,
we have $e_h = |S^0|$.
$f_h$ is the number of actions $a^{w+1}$ in $h$ such that
$a^{w+1}$ is not in $S^0$
but $s_{Z(h')}$ is not completely mixed at $h'$ either.
Finally,
\[
c_h = \prod_{\substack{0\leq w\leq K-1\\ a^{w+1}\notin S^0}} s_{Z(a^1,\dots, a^w)}(a^1,\dots, a^w)(a^{w+1}) ~\cdot~ \prod_{\substack{0\leq w\leq K-1\\ a^{w+1}\in S^0}} \frac{1}{|A^0(a^1,\dots, a^w)|},
\]
where the second term is defined to be 1 if $S^0 = \emptyset$.
Note that $\prob{h \mbox{ occurs under } s^\epsilon}>0$ for every $h\in I_i$.

The probability that the information set $I_i$ is reached under $s^\epsilon$ is
$\mathcal{P}(I_i) \triangleq  \sum_{h \in I_i} \prob{h \mbox{ occurs under } s^\epsilon} = \sum_{h\in I_i} c_h \epsilon^{e_h} (1-\epsilon)^{f_h}>0$.
Then
$\mathcal{P}(I_i)$ can be written as a polynomial in $\epsilon$, that is, $\mathcal{P}(I_i) = b_0 + b_1 \epsilon + b_2 \epsilon^2 + \ldots + b_r \epsilon^r$, where
the coefficients $b_0, \ldots, b_r$ may be zero, positive or negative.
Following Bayes' rule, for any history $h\in I_i$,
$$\mu^\epsilon_i(I_i)(h) = \frac{c_h \epsilon^{e_h} (1-\epsilon)^{f_h}}{\mathcal{P}(I_i)} =
\frac{c_h \epsilon^{e_h} (1-\epsilon)^{f_h}}{ b_0 + b_1 \epsilon + b_2 \epsilon^2 + \ldots + b_r \epsilon^r}>0.$$

To define the belief system $\mu$, let $d$ be the minimum degree of $\epsilon$ in $\mathcal{P}(I_i)$ such that
$b_d\neq 0$.
As the minimum degree of $\epsilon$ in each
term $c_h \epsilon^{e_h} (1-\epsilon)^{f_h}$ is $e_h$ with coefficient $c_h>0$,
we have $d = \min_{h\in I_i} e_h$ and $b_d = \sum_{h\in I_i, e_h=d} c_h >0$.
For any $h\in I_i$, we define $\mu_i(I_i)(h) = {c_h}/{b_d}(>0)$ if $e_h =d$, and $\mu_i(I_i)(h) = 0$ if $e_h > d$.
It is easy to see that $\mu_i(I_i)$ is a probability distribution over $I_i$.
Moreover, $\lim_{\epsilon \rightarrow 0} \mu^\epsilon_i(I_i)(h) = c_h/b_d$ when $e_h = d$, and
$\lim_{\epsilon \rightarrow 0} \mu^\epsilon_i(I_i)(h) = 0$ when $e_h>d$.
Thus,
$\lim_{\epsilon \rightarrow 0} \mu^\epsilon_i(I_i)(h) = \mu_i(I_i)(h)$ for any player $i$, information set $I_i$ of $i$ and history $h\in I_i$,
and $\mu^\epsilon$ converges to $\mu$ as $\epsilon\rightarrow 0$.
Since $s^{\epsilon}$ converges to $s$ as we have seen, $s$ and $\mu$
are consistent.
%

For sequential rationality, the only thing we need to show is that, at a reachable information set, the belief specified by $\mu$
is derived from $s$ using Bayes' rule.
To do so,
consider an arbitrary player $i$ and an information set $I_i$ of $i$ that is reachable by $s$.
By definition, there exists $h\in I_i$ such that $e_h = 0$, thus $d = 0$ for $\mathcal{P}(I_i)$ and $b_0 = \sum_{h\in I_i, e_h=0} c_h$.
Therefore $\mu_i(I_i)$ is indeed the probability distribution derived from $s$ using Bayes' rule.
Sequential rationality of $s$ (with respect to $\mu$) then follows from the definition of SSE.
%
%
Thus $({s}, \mu)$
is a sequential equilibrium.
\end{proof}

\paragraph*{Alternate definition of strong sequential equilibrium}
The notion of strong sequential equilibrium requires
that at any unreachable information set, regardless of the belief
the acting player holds at that set, his action should be a best response to that belief and the other players' strategies.
We now give an equivalent definition
of SSE, which
says that a player's
strategy
at an unreachable
information set should be optimal following {\em every history} in that information set.
This definition is more convenient when proving that a strategy profile is an SSE.

\begin{definition}
\deflabel{strong-se-2}
A strategy profile $s$ is a {\em strong sequential
equilibrium} if for every player $i$ and information set $I_i$ of $i$,
we have:
\begin{itemize}[noitemsep,nolistsep,leftmargin=*]
\item {\bf At reachable information sets $I$:} conditional on $I_i$ being reached,
player $i$'s strategy $s_i$ is a best response to $s_{-i}$, given $i$'s beliefs at $I_i$ being derived from
$s$ using Bayes' rule.

\item {\bf At unreachable information sets $I_i$:} for every history $h\in I_i$, conditional on $I_i$ being reached,
player $i$'s strategy $s_i$ is a best response to $s_{-i}$, given $i$'s belief that he is at $h$ with probability 1.
\end{itemize}
\end{definition}

We now prove the equivalence of the two definitions of SSE in the following lemma.
W.l.o.g., $s$ is a profile of pure strategies.

\begin{lemma}
\lemlabel{unreachable}
For any strategy profile ${s}$, any player $i$ and information set $I_i$ of $i$ that is not reached with positive probability under ${s}$,
conditional on $I_i$ being reached,
$s_i$ is a best response to $s_{-i}$ with respect to all possible beliefs that player $i$ may hold at $I_i$
if and only if for every history $h\in I_i$,
$s_i$ is a best response
to $s_{-i}$ given $i$'s belief that he is at $h$ with probability 1.
\end{lemma}

\begin{proof}
The ``only if'' part is immediate, because for any history $h\in I_i$, ``at $h$ with probability~1 (and any other history with probability 0)'' is a specific belief that $i$ may hold at $I_i$.

The ``if'' part is also easy to show.
Suppose that $s_i$ is a best response to $s_{-i}$ conditional on every history $h \in I_i$ (i.e., at $h$ with probability 1).
To show that $s_i$ is
a best response to $s_{-i}$ conditional on all possible beliefs
player $i$ may hold at information set $I_i$,
arbitrarily fix
a belief $\mu_i(I_i)$ over $I_i$
and a strategy $s_i'$.
Let $I_i = \{h_1, h_2, \ldots, h_m\}$
and $\mu_i(I_i) = (\mu_i(I_i)(h_1), \mu_i(I_i)(h_2),
\ldots, \mu_i(I_i)(h_m))$,
 where $\mu_i(I_i)(h_k)$ is the probability with which player $i$ believes that history $h_k$ occurs conditional on $I_i$ being reached.
Then, player $i$'s expected utilities under $s_i$ and $s_i'$ respectively,
conditioned on $I_i$, $\mu_i(I_i)$ and $s_{-i}$, are
\[ u_i(s_i, s_{-i}|\mu_i(I_i)) = \sum_{k =1}^m \mu_i(I_i)(h_k) \cdot u_i( s_i, s_{-i} | h_k) \mbox{ and } u_i(s'_i, s_{-i}|\mu_i(I_i)) = \sum_{k =1}^m \mu_i(I_i)(h_k) \cdot u_i( s_i', s_{-i} | h_k),\]
where $u_i( s_i, s_{-i} | h_k)$ is player $i$'s utility under $(s_i, s_{-i})$, conditioned on history $h_k$
being reached at $I_i$.
Since $s_i$ is a best response to $s_{-i}$ at every $h_k \in I_i$, we have
$u_i(s_i, s_{-i} | h_k) \geq u_i(s'_i, s_{-i} |h_k) \ \forall k \in \{1, \ldots, m\}.$
%
Thus $u_i(s_i, s_{-i}|\mu_i(I_i)) \geq u_i(s'_i, s_{-i}|\mu_i(I_i))$ and the ``if'' part holds.
\end{proof}

\paragraph*{One-shot deviation for strong sequential equilibrium}
Informally, the one-shot deviation principle says that a player cannot
change his action at a single information set (without changing the
rest of his strategy) and improve his expected reward.

In the context of sequential equilibrium, it is well
known that given a consistent belief system $\mu$, $(s,
\mu)$ is a sequential equilibrium if and only if the {\em one-shot deviation
principle holds}, that is, no player $i$ has an information set $I_i$ at which
a change in $s_i(I_i)$---holding the remainding of $s_i$ fixed---increases
his expected utility conditional on reaching $I_i$~\cite{osborne1994course,
hendon1996one}.

Since strong sequential equilibrium does not require artificial notion of beliefs for unreachable information sets,
we define a stronger
notion of one-shot deviation at those information sets---
for every decision node (i.e., history) in an unreachable information set of player $i$, there does not
exist a one-shot deviation at that node which improves player $i$'s utility conditional on that node being reached. Note that at reachable
information sets, both the definition and proof of the one-shot deviation condition for SSE are exactly the same as in SE~\cite{hendon1996one}.


\begin{lemma}[One-shot deviation for strong sequential equilibrium]
\lemlabel{one-shot}
For any strategy profile~${s}$, ${s}$ is a strong sequential equilibrium if and only if
it satisfies the following {\em one-shot deviation principle}:
For every player $i$ and every information set $I_i$ of $i$,

\begin{itemize}[noitemsep,nolistsep,leftmargin=*]
\item {\bf If $I_i$ is reachable under ${s}$:} there does not exist a change in $s_i(I_i)$ (holding
the rest of $s_i$ fixed) that increases player $i$'s expected utility conditional on reaching $I_i$,
given his belief at $I_i$ derived using
Bayes' rule.
\item {\bf If $I_i$ is unreachable under ${s}$:} for every history $h \in I_i$, there does not exist a change in $s_i(I_i)$ (holding the rest of
$s_i$ fixed) that increases player $i$'s expected utility conditional on reaching $h$.
\end{itemize}
\end{lemma}

\begin{proof}
The ``only if'' part follows immediately from \defref{strong-se-2} and the fact that a one-shot deviation results in a different strategy for the deviating player.
We now prove the ``if'' part,
that is, if $s$ satisfies the one-shot deviation principle then it is a strong sequential equilibrium.

{\bf Reachable information sets.}
First of all, similar to the proof of \lemref{lem:SE},
we can construct a belief system $\mu$ such that $s$ and $\mu$ are consistent.
Indeed, the construction of $\mu$ only depends on the actions taken by $s$ and does not depend on the utilities induced by $s$ at all.
Since $s$ satisfies the one-shot deviation principle at every reachable information set and at every history in each unreachable information set,
it is not hard to see that $s$ satisfies the one-shot deviation principle with respect to $\mu$.
Thus $(s, \mu)$ is a sequential equilibrium.
Accordingly, for any player $i$ and information set $I_i$ of $i$ that is reachable by $s$,
$s_i$ is a best response to $s_{-i}$ conditional on $\mu_i(I_i)$ (which is derived from $s$ using Bayes' rule at $I_i$),
as desired by the definition of SSE.



{\bf Unreachable information sets.}
Next, we use backward induction to show that, for any player $i$, information set $I_i$ of $i$ that is unreachable by $s$,
and history $h\in I_i$,
$s_i$ is a best response to $s_{-i}$ conditional on reaching $h$.
To begin with,
if $h$ is of height 1 then this immediately holds: indeed, the strategy induced by $s_i$ following $h$ is exactly the action $s_i(I_i)$,
thus the one-shot deviation principle implies that $s_i$ is a best response to $s_{-i}$ at $h$.

Now, arbitrarily fix a player $i$, information set $I_i$ of $i$ unreachable by $s$, and a history $h\in I_i$ of height larger than 1.
By induction, assume that for any
information set $I'_i$ of $i$ unreachable by $s$, and history $h'\in I'_i$ of height smaller than that of $h$,
$s_i$ is a best response to $s_{-i}$ at $h'$.
%
For the sake of contradiction, suppose
player $i$ can deviate to strategy $s_i'$ and increase his utility conditional on reaching $h$,
that is,
$$u_i(s_i', s_{-i}|h)> u_i(s_i, s_{-i}|h).$$

If $s'_i(I_i)=s_i(I_i)$, consider the first history $h'$ following $h$ where player $i$ acts and $s_i'$ differs from $s_i$.
As $h$ is unreachable by $s$, $h'$ is unreachable by $s$ as well.
However, the height of $h'$ is smaller than that of $h$ and $u_i(s_i', s_{-i}|h') = u_i(s_i', s_{-i}|h)> u_i(s_i, s_{-i}|h) = u_i(s_i, s_{-i}|h')$,
contradicting the inductive hypothesis.
Thus we have
$$s'_i(I_i)\neq s_i(I_i).$$

If $s'_i$ is the same as $s_i$ at all the histories following $(h, s'_i(I_i))$ where player $i$ acts,
then the one-shot deviation principle is violated.
Accordingly, there must exist a history following $(h, s'_i(I_i))$, where player $i$ acts and $s'_i$ differ from $s_i$.
Letting $h'$ be the first such history, we have that the height of $h'$ is smaller than that of $h$.
Since $h'$ is unreachable by $s$, by the inductive hypothesis we have that $s_i$ is a best response to $s_{-i}$ at $h'$. Thus
$u_i(s_i, s_{-i}|h')\geq u_i(s'_i, s_{-i}|h')$.
As $u_i(s'_i, s_{-i}|h') = u_i(s'_i, s_{-i}|h)> u_i(s_i, s_{-i}|h)$,
we have
$$u_i(s_i, s_{-i}|h')> u_i(s_i, s_{-i}|h).$$

Let strategy $s''_i$ be such that, it follows $s_i$ till history $h$, then follows
action $s'_i(I_i)$, then follows $s'_i$ (and $s_i$ as well, because they are the same after $(h, s'_i(I_i))$ and before $h'$)
till history $h'$,
 and then follows $s_i$ for the rest.
 Note that $s''_i$ can be obtained from $s_i$ by a one-shot deviation from $s_i(I_i)$ to $s'_i(I_i)$.
However,
$$u_i(s''_i, s_{-i}|h) = u_i(s''_i, s_{-i}|h') = u_i(s_i, s_{-i}|h')> u_i(s_i, s_{-i}|h),$$
contradicting the one-shot deviation principle.
Therefore $s_i$ is a best response to $s_{-i}$ conditional on reaching $h$, as desired.

Combining everything together, by~\defref{strong-se-2},
 ${s}$ is an SSE and \lemref{one-shot} holds.
\end{proof}

\paragraph*{Verifying strong sequential equilibrium}
Given an extensive-form game with arbitrary number of players, it is possible to decide whether 
a pair $(s, \mu)$ is a sequential equilibrium in time polynomial in the size of the game tree~\cite{gatti2012new}.

However, if only a strategy profile $s$ is given, then it is NP-hard to decide whether $s$ is part of an SE (that is, whether there exists a belief system
$\mu$ such that $(s, \mu)$ is an SE) \cite{hansen2010computational}.
As strong sequential equilibrium does not rely on belief systems, we prove the following.

\begin{lemma}
\lemlabel{verify-sse}
Given an extensive-form game and a strategy profile ${s}$ of the players, deciding whether
${s}$ is a SSE of the game can be done in time polynomial in the size of the game tree.
\end{lemma}

\begin{proof}
First of all, we can traverse the game tree in polynomial time,
 mark each information set whether it is reachable by $s$ or not, 
 and compute, for each player $i$ and each reachable information set $I_i$ of $i$,
the belief $\mu_i(I_i)$ derived from $s$ using Bayes' rule.
Next, we apply the one-shot deviation principle following \lemref{one-shot}.

To do so, we start from the bottom level of the tree and proceed up.
For every player $i$ and every information set $I_i$ of $i$, 
if $I_i$ is unreachable under $s$, then we go
through each $h \in I_i$ and each $a \in A(I_i)$, and check if changing $s_i(I_i)$ to $a$ improves $i$'s
utility conditional on reaching $h$. 
If so then 
$s$ is not an SSE.
If $I_i$ is reachable under $s$, then we go through every $a \in A(I_i)$, and
check if changing $s_i(I_i)$ to $a$ improves $i$'s expected utility conditional on $I_i$ and $\mu_i(I_i)$.
If so then again 
$s$ is not an SSE.
If all the checks above pass, then $s$ is an SSE.

Since this procedure goes through each decision node of the game tree at most once, and since it takes polynomial time to compute player $i$'s (expected) utility under $s$
following a decision node (or an information set), 
deciding whether $s$ is an SSE takes polynomial time in the size of the tree.
%
\end{proof}

\bibliographystyle{abbrv}

%

\end{document}